\documentclass[11pt]{article}
\usepackage{graphicx}
\usepackage{tabularx}
\usepackage{url}
\usepackage{amsthm}
\usepackage{amsmath}
\usepackage{amsfonts}
\usepackage{amssymb}

\theoremstyle{plain}
\newtheorem{theorem}{Theorem}
\newtheorem{lemma}{Lemma}
\newtheorem{proposition}{Proposition}

\newtheorem{corollary}{Corollary}
\theoremstyle{definition}

\theoremstyle{remark}

\date{}

\begin{document}

\title{Rough analysis of computation trees}
\author{Mikhail Moshkov\thanks{Computer, Electrical and Mathematical Sciences and Engineering Division,
King Abdullah University of Science and Technology (KAUST),
Thuwal 23955-6900, Saudi Arabia. Email: mikhail.moshkov@kaust.edu.sa.
}}
\maketitle

\begin{abstract}
This paper deals with  computation trees over an arbitrary
structure consisting of a set along with collections of functions and   predicates that are defined on it. It is devoted to the  comparative analysis of three parameters of problems with $n$ input variables over this structure: the complexity of a problem description, the minimum complexity of a computation tree solving this problem deterministically, and the minimum complexity of a computation tree solving this problem nondeterministically. Rough classification of relationships among
these parameters is considered and all possible seven types of these
relations are enumerated. The changes of relation types with the growth of the number $n$ of input variables are studied.
\end{abstract}


\section{Introduction}
\label{S1}

Computation trees are well known models of algorithms. They are a
natural generalization of decision trees: besides one-place operations of
predicate type (attributes) which are used in decision trees, in
computation trees many-place predicate and functional operations may be used.
Just as in general algorithm theory, where both deterministic and nondeterministic
algorithms are considered, it is expedient to
study not only deterministic but also nondeterministic computation trees.

Linear decision trees and algebraic decision and computation trees were studied most intensively.
Lower bounds on the complexity were obtained in \cite{Bjorner92,Dobkin78,Dobkin79} for linear decision trees, in \cite{Steele82,Yao95,Yao97} for algebraic decision trees, in \cite{Ben-Or83,Gabrielov17} for algebraic computation trees, and in  \cite{Grigoriev96} for Pfaffian computation trees. Upper bounds on the complexity were obtained in \cite{Dobkin76,Heide84,Moshkov82} for linear decision trees and in \cite{Moshkov05} for quasilinear decision trees that includes linear decision trees and some kinds of algebraic decision trees. Nondeterministic linear decision and computation trees were studied in \cite{Heide85} and \cite{Moshkov88,Moshkov05}, respectively.

The complexity of deterministic decision trees over arbitrary infinite sets of $k$-valued attributes, $k \ge 2$, was studied in \cite{Moshkov03,Moshkov05}. Relationships between deterministic and different kinds of nondeterministic decision trees over arbitrary infinite sets of $k$-valued attributes were investigated in \cite{Moshkov96,Moshkov20}.

In this paper, we study computation trees over an arbitrary structure $U=(A,F,P)$
consisting of a set $A$ along with a collection of functions $F$ and a collection of
predicates $P$ that are defined on it.

For each natural $n$, we describe a set $%
\mathcal{P}(U,n)$ of problems over $U$ with $n$ input variables. Each such problem with input variables $x_1, \ldots , x_n$ is given by a finite sequence of functional and predicate expressions over $U$. This sequence  defines $r$ functions $\alpha_1, \ldots, \alpha_r$ of the form $p(t_{1},\ldots ,t_{m})$, where $p \in P$ and
$t_{1},\ldots ,t_{m}$ are functions with variables from the set $\{x_1, \ldots , x_n\}$ obtained
from functions contained in $F\cup \{x\}$ by the operation of substitution.
The functions $\alpha_1, \ldots, \alpha_r$ divide the set $A^n$ into areas in which these functions are constant. Each area is labeled with a finite nonempty set of solutions. For a given $n$-tuple $\bar{a} \in A^n$, we should find a solution from the set attached to the area to which $\bar{a}$ belongs. Various problems of combinatorial
optimization, pattern recognition, computational geometry, etc., can be represented in such form.

We define a
complexity measure $\psi $ and, for each problem $z \in \mathcal{P}(U,n)$, we consider three
parameters: $\psi _{U}^{i}(z)$ -- the complexity of the problem $z$
description, $\psi _{U}^{d}(z)$ -- the minimum complexity of a computation
tree that solves the problem $z$ deterministically, and $\psi _{U}^{a}(z)$
-- the minimum complexity of a computation tree that solves the problem $z$
nondeterministically. The pair $(U,\psi )$ is called a sm-pair
((structure, measure)-pair).

To study relationships among these parameters, for each $b,c\in \{i,d,a\}$,
we consider two partial functions defined on the set of nonnegative integers:%
\begin{eqnarray*}
\mathcal{U}_{U\psi n}^{bc}(m) &=&\max \{\psi _{U}^{b}(z):z\in \mathcal{P}%
(U,n),\psi _{U}^{c}(z)\leq m\}, \\
\mathcal{L}_{U\psi n}^{bc}(m) &=&\min \{\psi _{U}^{b}(z):z\in \mathcal{P}%
(U,n),\psi _{U}^{c}(z)\geq m\}.
\end{eqnarray*}%
If the value $\mathcal{U}_{U\psi n}^{bc}(m)$ is defined for some $m$, then
it is the unimprovable upper bound on the value $\psi _{U}^{b}(z)$ for
problems $z\in \mathcal{P}(U,n)$ such that the inequality $\psi
_{U}^{c}(z)\leq m$ holds. If the value $\mathcal{L}_{U\psi n}^{bc}(m)$ is
defined for some $m$, then it is the unimprovable lower bound on the value $%
\psi _{U}^{b}(z)$ for problems $z\in \mathcal{P}(U,n)$ such that the
inequality $\psi _{U}^{c}(z)\geq m$ holds.

The study of the functions $\mathcal{U}_{U\psi n}^{bc}$ and $\mathcal{L}%
_{U\psi n}^{bc}$ directly is, in general case, too complicated problem.
Therefore, instead of the functions $\mathcal{U}_{U\psi n}^{bc}$ and $\mathcal{L}%
_{U\psi n}^{bc}$, we study their types $typ(\mathcal{U}_{U\psi n}^{bc})$ and
$typ(\mathcal{L}_{U\psi n}^{bc})$ from the set $\{\alpha ,\beta ,\gamma
,\delta ,\varepsilon \}$. In particular, the type $\alpha $ means that the considered
function $\mathcal{U}_{U\psi n}^{bc}$ has infinite domain and is bounded
from above, types $\beta $, $\gamma $, and $\delta $ describe different
variants of the growth of the unbounded function $\mathcal{U}_{U\psi n}^{bc}$
with an infinite domain, and the type $\varepsilon $ means that the function
$\mathcal{U}_{U\psi n}^{bc}$ has a finite domain.

All pairs $typ(\mathcal{L}_{U\psi n}^{bc})$ $typ(\mathcal{U}_{U\psi n}^{bc})$%
, $b,c\in \{i,d,a\}$, form the $n$-type of the sm-pair $%
(U,\psi )$ that is the table $typ(U,\psi ,n)$ with three rows and three
columns in which rows from top to bottom and columns from the left to the
right are labeled with indices $i,d,a$, and the pair $typ(\mathcal{L}_{U\psi
n}^{bc})$ $typ(\mathcal{U}_{U\psi n}^{bc})$ is in the intersection of the
row with index $b\in \{i,d,a\}$ and the column with index $c\in \{i,d,a\}$.
We describe all possible seven $n$-types of sm-pairs. These
results are similar to ones obtained for decision trees \cite{Moshkov96}.

For the sm-pair $(U,\psi )$, $n$-types can change with the
increasing of $n$. To investigate this phenomenon, we study the infinite
sequence $$type(U,\psi )=typ(U,\psi ,1)typ(U,\psi ,2)\cdots $$ that is called
the dynamic type of the sm-pair $(U,\psi )$. In this paper,
we describe all possible dynamic types of  sm-pairs.

Some preliminary results in this direction were published without proofs in
\cite{Moshkov97a,Moshkov97b,Moshkov98a}. The publication of the final results
with proofs was postponed for years: only in the present paper it was possible
to finally resolve the issue of the structure of the set of all possible dynamic
types of  sm-pairs. This required considering of non-trivial
constructions in Section \ref{S6}. This section contains also examples of study of different sm-pairs.

The rest of the paper is organized as follows. In Sections \ref{S2} and \ref{S3}, basic notions and main results are considered. Sections \ref{S4}--\ref{S7} are devoted to the proofs of auxiliary statements and main theorems. Section \ref{S8} contains some explanations of the results of this paper and Section \ref{S9} -- short conclusions.

\section{Basic Notions}
\label{S2}

In this section, we consider the notions of structure, computation tree, problem, complexity measure,  sm-pair ((structure,measure)-pair), type of function, $n$-type of sm-pair, and dynamic type of sm-pair.

\subsection{Structures}

Let $\omega =\{0,1,2,\ldots \}$ be the set of nonnegative integers, $%
E_{2}=\{0,1\}$, and $X=\{x_{i}:i\in \omega \}$ be the set of variables. Let $%
A$ be a nonempty set, $F$ be a set of functions of the kind $f(x_{1},\ldots
,x_{n})$, where $f:A^{n}\rightarrow A$ and $n\in \omega $ (if $n=0$, then $f$
is a constant), $P$ be a nonempty set of predicates (relations) of the kind $%
p(x_{1},\ldots ,x_{n})$, where $p:A^{n}\rightarrow E_{2}$ and $n\in \omega
\setminus \{0\}$, and $F\cap P=\emptyset $. The triple $U=(A,F,P)$ is called
a structure.

We denote by $[F]$ the set of all functions with variables from $X$ obtained
from functions contained in $F\cup \{x\}$ by the operation of substitution. We
denote by $P[F]$ the set of all functions of the kind $p(t_{1},\ldots ,t_{n})$, where
$t_{1},\ldots ,t_{n}\in \lbrack F]$ and $p$ is a predicate from $P$ with $n$
variables.

An expression $x_{j}\Leftarrow f(x_{l_{1}},\ldots ,x_{l_{r}})$, where $f$ is
a function from $F$ with $r$ variables, will be called a functional
expression over $U$. An expression $p(x_{q_{1}},\ldots ,x_{q_{k}})$, where $%
p $ is a predicate from $P$ with $k\geq 1$ variables, will be called a
predicate expression over $U$.

Consider a pair $(Y,\beta )$, where $Y$ is a finite nonempty subset of the
set $X$ and $\beta $ is a finite sequence of functional and predicate
expressions over $U$. We now correspond to each predicate expression from $%
\beta $ a function from $P[F]$. Let $\beta =\beta _{1},\ldots ,\beta _{m}$
and $s$ be the minimum number from $\omega $ such that all variables from $Y$
and all variables from the expressions $\beta _{1},\ldots ,\beta _{m}$ are
contained in the set $\{x_{0},\ldots ,x_{s}\}$. Let $x_{w}$ be the variable
from the set $Y$ with the minimum index $w$.

For $i=0,\ldots ,m$, we define a sequence $t_{i}=t_{i0},\ldots ,t_{is}$ of
functions from $[F]$ with variables from $Y$. Let $j\in \{0,\ldots ,s\}$. If
$x_{j}\in Y$, then $t_{0j}=x_{j}$. If $x_{j}\notin Y$, then $t_{0j}=x_{w}$.
Let sequences $t_{0},\ldots ,t_{i}$, $0\leq i<m$, be already defined. If $%
\beta _{i+1}$ is a predicate expression, then $t_{i+1}=t_{i}$. If $\beta
_{i+1}$ is a functional expression $x_{j}\Leftarrow f(x_{l_{1}},\ldots
,x_{l_{n}})$, then $t_{i+1}=t_{i0},\ldots ,t_{ij-1},f(t_{il_{1}},\ldots
,t_{il_{n}}),t_{ij+1},\ldots ,t_{is}$.

Let there be exactly $r>0$ predicate expressions $\beta _{c_{1}},\ldots
,\beta _{c_{r}}$ among $\beta _{1},\ldots ,$ $\beta _{m}$, where $c_{1}<\cdots
<c_{r}$. For $i=1,\ldots ,r$, we associate with the expression $\beta
_{c_{i}}$ a function $\alpha _{i}\in P[F]$ with variables from $Y$. Let $%
\beta _{c_{i}}$ be an expression $p(x_{q_{1}},\ldots ,x_{q_{k}})$. Then $%
\alpha _{i}=p(t_{c_{i}q_{1}},\ldots ,t_{c_{i}q_{k}})$. We denote by $\Pi
(Y,\beta )$ the $r$-tuple $(\alpha _{1},\ldots ,\alpha _{r})$.

\subsection{Computation Trees}

A node in a finite directed tree is called the root, if it is the only node
without entering edges. A tree, which has such a node, is called a finite
directed tree with the root. The tree nodes without leaving edges are called
terminal nodes. The tree nodes, which are neither the root nor terminal, will be
called working nodes. A complete path in a finite directed tree with the root is
any sequence $\xi =v_{0},d_{0},\ldots ,v_{m},d_{m},v_{m+1}$ of nodes and
edges of the tree such that $v_{0}$ is the root, $v_{m+1}$ is a terminal
node, and the edge $d_{i}$ leaves the node $v_{i}$ and enters the node $%
v_{i+1}$ for $i=0,\ldots ,m$.

A computation tree over the structure $U=(A,F,P)$ is a pair $\Gamma =(Y,G)$, where $%
Y $ is a finite nonempty subset of the set of variables $X$ and $G$ is a
marked finite directed tree with the  root, which has at least two nodes and
satisfies the following conditions:
\begin{itemize}
\item The root and the edges leaving the root are not labeled.

\item Each working node is a functional or a predicate node.

\item Each functional node is labeled with a functional expression over $U$,
and each edge leaving a functional node is not labeled.

\item Each predicate node is labeled with a predicate expression over $U$, and
each edge leaving a predicate node is labeled with a number from $E_{2}$.

\item Each terminal node is labeled with a number from $\omega $.
\end{itemize}
A computation tree is called deterministic if it satisfies the following
conditions:
\begin{itemize}
\item There is exactly one edge leaving the root.{}

\item Each functional node has exactly one edge leaving it.

\item For each predicate node, edges leaving this node are labeled with
pairwise different numbers from $E_{2}$.
\end{itemize}

The set of computation trees over the structure $U$ will be denoted by $\mathcal{T}%
(U)$. Let $\Gamma =(Y,G)\ $be a computation tree over $U$. Nodes, edges and
paths in the tree $G$ are called nodes, edges and paths in the computation tree $%
\Gamma $. The set $Y$ is called the set of input variables for computation tree $%
\Gamma $. Let $Y=\{x_{l_{1}},\ldots ,x_{l_{n}}\}$ and $l_{1}\,<\cdots <l_n$.
Denote $\bar{x}=(x_{l_{1}},\ldots ,x_{l_{n}})$. We denote by $\Xi (\Gamma )$
the set of complete paths in $\Gamma $. Let $\xi =v_{0},d_{0},\ldots
,v_{m},d_{m},v_{m+1}$ be a complete path in $\Gamma $. We denote by $\kappa
(\xi )$ the number assigned to the node $v_{m+1}$. We now define a sequence $%
\beta (\xi )$ of functional and predicate expressions over $U$ and a subset $%
\mathcal{A}(\xi )$ of the set $A^{n}$ associated with $\xi $. If $m=0$, then
$\beta (\xi )$ is the empty sequence. Let $m>0$, and let the expression $%
\beta _{j}$ be assigned to the node $v_{j}$, $j=1,\ldots ,m$. Then $\beta
(\xi )=\beta _{1},\ldots ,\beta _{m}$. If there are no predicate expressions
in the sequence $\beta (\xi )$, then $\mathcal{A}(\xi )=A^{n}$. Let there are
exactly $r>0$ predicate expressions $\beta _{c_{1}},\ldots ,\beta _{c_{r}}$
among $\beta _{1},\ldots ,\beta _{m}$, where $c_{1}<\cdots <c_{r}$. Let $\Pi
(Y,\beta )=(\alpha _{1},\ldots ,\alpha _{r})$, and let $\delta _{i}$ be the
number assigned to the edge $d_{c_{i}}$, $i=1,\ldots ,r$. Then $\mathcal{A}%
(\xi )$ is the set of solutions on $A^{n}$ for the system of equations
\[
\{\alpha _{1}(\bar{x})=\delta _{1},\ldots ,\alpha _{r}(\bar{x})=\delta
_{r}\}.
\]

\subsection{Problems}

The set of nonempty finite subsets of the set $\omega $ will be denoted by $%
S(\omega )$. A problem over the structure $U$ is any tuple of the kind $z=(Y,\nu
,\beta _{1},\ldots ,\beta _{m})$, where $Y$ is a nonempty finite subset of
the set $X$, $m\in \omega \setminus \{0\}$, $\beta _{1},\ldots ,\beta _{m}$
are functional and predicate expressions over $U$, and there is $r\in \omega
\setminus \{0\}$ such that $\nu :E_{2}^{r}\rightarrow S(\omega )$ and there
are exactly $r$ predicate expressions in the sequence $\beta _{1},\ldots
,\beta _{m}$. The set $Y$ is called the set of input variables for the
problem $z$. We denote by $\beta (z)$ the sequence $\beta _{1},\ldots ,\beta
_{m}$. Let $|Y|=n$ and $\Pi (Y,\beta (z))=(\alpha _{1},\ldots ,\alpha _{r})$%
. The problem $z$ may be interpreted as a problem of searching for at least
one number from the set $z(\bar{a})=\nu (\alpha _{1}(\bar{a}),\ldots ,\alpha
_{r}(\bar{a}))$ for an arbitrary $\bar{a}\in A^{n}$. Different problems of
pattern recognition, combinatorial optimization, and computational geometry
can be represented in such form. We denote by $\mathcal{P}(U)$ the set of
problems over the structure $U$.

Let $z=(Y_{1},\nu ,\beta _{1},\ldots ,\beta _{m})\in \mathcal{P}(U)$ and $%
\Gamma =(Y_{2},G)\in \mathcal{T}(U)$. Let $|Y_{1}|=n$. We will say that the
computation tree $\Gamma $ solves the problem $z$ nondeterministically if the
following conditions hold:
\begin{itemize}
 \item $Y_{1}=Y_{2}$.
 \item $\bigcup_{\xi \in \Xi (\Gamma )}\mathcal{A}(\xi )=A^{n}$.
 \item For any $\bar{a}\in A^{n}$ and any $\xi \in \Xi (\Gamma )$ such that $%
\bar{a}\in \mathcal{A}(\xi )$, the relation $\kappa (\xi )\in z(\bar{a})$
holds.
 \end{itemize}
We will say that the computation tree $\Gamma $ solves the problem $z$ deterministically if $%
\Gamma $ is a deterministic computation tree, which solves $z$
nondeterministically.

\subsection{Complexity Measures and SM-Pairs}

Let $U=(A,F,P)$ be a structure. Denote by $(F\cup P)^{\ast }$ the set of all
finite words over the alphabet $F\cup P$, including the empty word $\lambda $%
. A complexity measure over the structure $U$ is any map of the kind $\psi
:(F\cup P)^{\ast }\rightarrow \omega $. The complexity measure $\psi $ is
called limited if it has the following properties:
\begin{itemize}
\item $\psi (\alpha _{1}\alpha _{2})\leq \psi (\alpha _{1})+\psi (\alpha _{2})$
for any $\alpha _{1},\alpha _{2}\in (F\cup P)^{\ast }$.

\item $\psi (\alpha _{1}\alpha _{2}\alpha _{3})\geq \psi (\alpha _{1}\alpha
_{3})$ for any $\alpha _{1},\alpha _{2},\alpha _{3}\in (F\cup P)^{\ast }$.

\item For any $\alpha \in (F\cup P)^{\ast }$, the inequality $\psi (\alpha
)\geq |\alpha |$ holds, where $|\alpha |$ is the length of $\alpha $.
\end{itemize}

We extend the complexity measure $\psi $ onto the set of all finite sequences of
functional and predicate expressions over $U$ in the following way: $\psi
(\beta )=\psi (\lambda )$ if $\beta $ is the empty sequence. Let $\beta $ be
a nonempty sequence and $\beta =\beta _{1},\ldots ,\beta _{m}$. Then $\psi
(\beta )=\psi (\alpha )$, where $\alpha =b_{1}\cdots b_{m}\in (F\cup
P)^{\ast }$ and, for $i=1,\ldots ,m$, if $\beta _{i}$ is a predicate
expression $p(x_{q_{1}},\ldots ,x_{q_{k}})$, then $b_{i}=p$, and if $\beta
_{i}$ is a functional expression $x_{j}\Leftarrow f(x_{l_{1}},\ldots
,x_{l_{r}})$, then $b_{i}=f$.

We extend the complexity measure $\psi $ onto the set $\mathcal{T}(U)$ of
computation trees over $U$ as follows: $\psi (\Gamma )=$ $\max \{\psi
(\beta (\xi )):\xi \in \Xi (\Gamma )\}$ for any $\Gamma \in \mathcal{T}(U)$.
The value $\psi (\Gamma )$ will be called the $\psi $-complexity of a
computation tree $\Gamma $.

We now consider some examples of complexity measures. Let $w:(F\cup
P)\rightarrow \omega \setminus \{0\}$. We define the function $\psi
^{w}:(F\cup P)^{\ast }\rightarrow \omega $ in the following way: for any $%
\alpha \in (F\cup P)^{\ast }$, $\psi ^{w}(\alpha )=0$ if $\alpha =\lambda $,
and $\psi ^{w}(\alpha )=\sum_{i=1}^{m}w(b_{i})$ if $\alpha =b_{1}\cdots
b_{m} $. The function $\psi ^{w}$ is a limited complexity measure over $U$
and is called a weighted depth. If $w\equiv 1$, then the function $\psi ^{w}$
is called the depth.

Let $\psi $ be a complexity measure over $U$ and $z\in \mathcal{P}(U)$. The
value $\psi _{U}^{i}(z)=\psi (\beta (z))$ is called the complexity of the
problem $z$ description. We denote by $\psi _{U}^{d}(z)$ the minimum $\psi $%
-complexity of a computation tree $\Gamma \in \mathcal{T}(U)$, which solves the
problem $z$ deterministically. We denote by $\psi _{U}^{a}(z)$ the minimum $%
\psi $-complexity of a computation tree $\Gamma \in \mathcal{T}(U)$, which
solves the problem $z$ nondeterministically.

A pair $(U,\psi )$, where $U$ is a structure and $\psi $ is a complexity measure
over $U$, will be called a (structure,measure)-pair or, in short, a sm-pair. If $\psi
$ is a limited complexity measure, then the pair $(U,\psi )$ will be called a
limited sm-pair.

\subsection{$n$-Types and Dynamic Types of SM-Pairs}

Let $(U,\psi )$ be a sm-pair and $n\in \omega \setminus \{0\}$. We denote by
$\mathcal{P}(U,n)$ the set of problems from $\mathcal{P}(U)$ with $n$ input
variables.

We have the three parameters $\psi _{U}^{i}(z)$, $\psi _{U}^{d}(z)$, and $%
\psi _{U}^{a}(z)$ for any problem $z\in P(U,n)$, and we investigate the
relationships between any two such parameters for problems from $\mathcal{P}%
(U,n)$. Let us consider, for example, the parameters $\psi _{U}^{i}(z)$ and $%
\psi _{U}^{d}(z)$. Let $m\in \omega $. We will study relations $\psi
_{U}^{i}(z)\leq m\Rightarrow \psi _{U}^{d}(z)\leq u$ true for any $z\in
\mathcal{P}(U,n)$. The minimum value of $u$ is most interesting for us. This
value (if exists) is equal to
\[
\mathcal{U}_{U\psi n}^{di}(m)=\max \{\psi _{U}^{d}(z):z\in \mathcal{P}%
(U,n),\psi _{U}^{i}(z)\leq m\}.
\]%
We also study relations $\psi _{U}^{i}(z)\geq m\Rightarrow \psi
_{U}^{d}(z)\geq l$. In this case, the maximum value of $l$ is most
interesting for us. This value (if exists) is equal to
\[
\mathcal{L}_{U\psi n}^{di}(m)=\min \{\psi _{U}^{d}(z):z\in \mathcal{P}%
(U,n),\psi _{U}^{i}(z)\geq m\}.
\]%
The two functions $\mathcal{U}_{U\psi n}^{di}$ and $\mathcal{L}_{U\psi n}^{di}$ describe how the
behavior of the parameter $\psi _{U}^{d}(z)$ depends on the behavior of the
parameter $\psi _{U}^{i}(z)$.

There are 18 similar functions for all ordered pairs of parameters $\psi
_{U}^{i}(z)$, $\psi _{U}^{d}(z)$, and $\psi _{U}^{a}(z)$. These 18 functions
well describe the relationships among the considered parameters. It would be
very interesting to enumerate 18-tuples of these functions for all sm-pairs.
But this is a very complicated problem.

In this paper, instead of functions
we study types of functions. With any function, we associate its type from
the set $\{\alpha ,\beta ,\gamma ,\delta ,\varepsilon \}$. For example, if a
function has infinite domain of definition, and it is bounded from above,
then its type is equal to $\alpha $. Thus, we enumerate 18-tuples of types
of functions. These tuples are represented as tables called the $n$-types of
sm-pairs. We also consider infinite sequences of the kind $1$-type of
sm-pair $2$-type of sm-pair ..., which are called dynamic types of sm-pairs
and characterize changes of relationships among the considered parameters
with the growth of the number of input variables.

We now give definitions of mentioned above notions. Let $b,c\in \{i,d,a\}$.
We define partial functions $\mathcal{U}_{U\psi n}^{bc}:\omega \rightarrow \omega $
and $\mathcal{L}_{U\psi n}^{bc}:\omega \rightarrow \omega $ as follows:
\[
\mathcal{U}_{U\psi n}^{bc}(m)=\max \{\psi _{U}^{b}(z):z\in \mathcal{P}%
(U,n),\psi _{U}^{c}(z)\leq m\},
\]%
\[
\mathcal{L}_{U\psi n}^{bc}(m)=\min \{\psi _{U}^{b}(z):z\in \mathcal{P}%
(U,n),\psi _{U}^{c}(z)\geq m\}.
\]%
If the value $\mathcal{U}_{U\psi n}^{bc}(m)$ is defined, then it is the
unimprovable upper bound on the value $\psi _{U}^{b}(z)$ for problems $z\in
\mathcal{P}(U,n)$ such that the inequality $\psi _{U}^{c}(z)\leq m$ holds.
If the value $\mathcal{L}_{U\psi n}^{bc}(m)$ is defined, then it is the
unimprovable lower bound on the value $\psi _{U}^{b}(z)$ for problems $z\in
P(U,n)$ such that the inequality $\psi _{U}^{c}(z)\geq m$ holds.

Let $g$ be a partial function from $\omega $ to $\omega $. We denote by $%
Dom(g)$ the domain of definition (domain for short) of $g$. Let $%
Dom^{+}(g)=\{n:n\in Dom(g),g(n)\geq n\}$ and $Dom^{-}(g)=\{n:n\in
Dom(g),g(n)\leq n\}$. We now define the value $typ(g)\in \{\alpha ,\beta
,\gamma ,\delta ,\varepsilon \}$, which will be called the type of the function $%
g $:

\begin{itemize}
\item If $Dom(g)$ is an infinite set and $g$ is a bounded from above
function, then $typ(g)=\alpha $.

\item If $Dom(g)$ is an infinite set, $Dom^{+}(g)$ is a finite set, and $g$
is an unbounded from above function, then $typ(g)=\beta $.

\item If each of the sets $Dom^{+}(g)$ and $Dom^{-}(g)$ is an infinite set,
then $typ(g)=\gamma $.

\item If $Dom(g)$ is an infinite set and $Dom^{-}(g)$ is a finite set, then $%
typ(g)=\delta $.

\item If $Dom(g)$ is a finite set, then $typ(g)=\varepsilon $.
\end{itemize}

To clarify the notion of a function type, consider in more detail the function $\mathcal{U}_{U\psi n}^{di}$.
One can show that $\psi
_{U}^{d}(z)\leq \psi _{U}^{i}(z)$ for any problem $z\in \mathcal{P}(U)$:
based on the problem description, it is easy to construct a computation tree,
which solves this problem deterministically and which complexity is equal to
the complexity of the problem description. Therefore $typ(\mathcal{U}_{U\psi n}^{di})\in
\{\alpha ,\beta ,\gamma \}$. If $typ(\mathcal{U}_{U\psi
n}^{di})=\alpha $, then there is a positive constant $p$ such that $%
\psi _{U}^{d}(z)\leq p$ for any problem $z\in \mathcal{P}(U,n)$. If $typ(%
\mathcal{U}_{U\psi n}^{di})=\gamma $, then there are infinitely many
numbers $m\in \omega $ for each of which there exists a problem $z\in \mathcal{P}%
(U,n)$ with $\psi _{U}^{d}(z)=\psi _{U}^{i}(z)=m$.  The case $typ(\mathcal{U}_{U\psi n}^{di})=\beta $ is
the most interesting for us: the function $\mathcal{U}_{U\psi
n}^{di}$ is not bounded from above and, for each problem with high enough complexity of
description, there exists a computation tree, which solves this problem
deterministically and which complexity is less than the complexity of the
problem description.

We denote by $typ(U,\psi ,n)$ a table with three rows and three columns, in
which rows from top to bottom and columns from the left to the right are
labeled with indices $i,d,a$, and the pair $typ(\mathcal{L}_{U\psi n}^{bc})$
$typ(\mathcal{U}_{U\psi n}^{bc})$ is in the intersection of the row with
index $b\in \{i,d,a\}$ and the column with index $c\in \{i,d,a\}$. The table
$typ(U,\psi ,n)$ is called the $n$-type of sm-pair $(U,\psi )$. We denote by
$dtyp(U,\psi )$ an infinite sequence $typ(U,\psi ,1)$ $typ(U,\psi ,2)\cdots $%
. The sequence $dtyp(U,\psi )$ is called the dynamic type of sm-pair $%
(U,\psi )$.

\section{Main Results}
\label{S3}

The main goal of this paper is to describe the set of all dynamic types of
sm-pairs and the set of all dynamic types of limited sm-pairs. The solution
of this problem allows us to characterize roughly all possible with the growth
of the number of input variables changes of relationships among the
complexity of problem description, the minimum complexity of computation trees
solving this problem deterministically, and the minimum complexity of
computation trees solving this problem nondeterministically.

Define seven tables:%
\[
T_{1}=%
\begin{tabular}{|c|ccc|}
\hline
& $i$ & $d$ & $a$ \\ \hline
$i$ & $\varepsilon \alpha $ & $\varepsilon \alpha $ & $\varepsilon \alpha $
\\
$d$ & $\varepsilon \alpha $ & $\varepsilon \alpha $ & $\varepsilon \alpha $
\\
$a$ & $\varepsilon \alpha $ & $\varepsilon \alpha $ & $\varepsilon \alpha $
\\ \hline
\end{tabular}%
\;\;T_{2}=%
\begin{tabular}{|c|ccc|}
\hline
& $i$ & $d$ & $a$ \\ \hline
$i$ & $\gamma \gamma $ & $\varepsilon \varepsilon $ & $\varepsilon
\varepsilon $ \\
$d$ & $\alpha \alpha $ & $\varepsilon \alpha $ & $\varepsilon \alpha $ \\
$a$ & $\alpha \alpha $ & $\varepsilon \alpha $ & $\varepsilon \alpha $ \\
\hline
\end{tabular}%
\
\]%
\[
T_{3}=%
\begin{tabular}{|c|ccc|}
\hline
& $i$ & $d$ & $a$ \\ \hline
$i$ & $\gamma \gamma $ & $\delta \varepsilon $ & $\varepsilon \varepsilon $
\\
$d$ & $\alpha \beta $ & $\gamma \gamma $ & $\varepsilon \varepsilon $ \\
$a$ & $\alpha \alpha $ & $\alpha \alpha $ & $\varepsilon \alpha $ \\ \hline
\end{tabular}%
\;\;T_{4}=%
\begin{tabular}{|c|ccc|}
\hline
& $i$ & $d$ & $a$ \\ \hline
$i$ & $\gamma \gamma $ & $\gamma \varepsilon $ & $\varepsilon \varepsilon $
\\
$d$ & $\alpha \gamma $ & $\gamma \gamma $ & $\varepsilon \varepsilon $ \\
$a$ & $\alpha \alpha $ & $\alpha \alpha $ & $\varepsilon \alpha $ \\ \hline
\end{tabular}%
\
\]%
\[
T_{5}=%
\begin{tabular}{|c|ccc|}
\hline
& $i$ & $d$ & $a$ \\ \hline
$i$ & $\gamma \gamma $ & $\gamma \varepsilon $ & $\gamma \varepsilon $ \\
$d$ & $\alpha \gamma $ & $\gamma \gamma $ & $\gamma \gamma $ \\
$a$ & $\alpha \gamma $ & $\gamma \gamma $ & $\gamma \gamma $ \\ \hline
\end{tabular}%
\;\;T_{6}=%
\begin{tabular}{|c|ccc|}
\hline
& $i$ & $d$ & $a$ \\ \hline
$i$ & $\gamma \gamma $ & $\gamma \varepsilon $ & $\gamma \varepsilon $ \\
$d$ & $\alpha \gamma $ & $\gamma \gamma $ & $\gamma \delta $ \\
$a$ & $\alpha \gamma $ & $\beta \gamma $ & $\gamma \gamma $ \\ \hline
\end{tabular}%
\
\]%
\[
T_{7}=%
\begin{tabular}{|c|ccc|}
\hline
& $i$ & $d$ & $a$ \\ \hline
$i$ & $\gamma \gamma $ & $\gamma \varepsilon $ & $\gamma \varepsilon $ \\
$d$ & $\alpha \gamma $ & $\gamma \gamma $ & $\gamma \varepsilon $ \\
$a$ & $\alpha \gamma $ & $\alpha \gamma $ & $\gamma \gamma $ \\ \hline
\end{tabular}%
\
\]

Let $T$ be a table and $i\in \omega $. We denote by $T^{i}$ the sequence $%
T\cdots T$, where $T$ is repeated $i$ times (if $i=0$, then $T^{i}$ is the
empty sequence). We denote by $T^{\infty }$ the infinite sequence $TTT\cdots
$. Denote $\Delta =\{T_{2}^{\infty },T_{2}^{i}T_{3}^{\infty
},T_{2}^{i}T_{3}^{j}T_{4}^{\infty }$, $T_{2}^{i}T_{3}^{j}T_{4}^{k}T_{7}^{%
\infty },T_{2}^{i}T_{5}^{\infty },T_{2}^{i}T_{5}^{j}T_{6}^{\infty }$, $%
T_{2}^{i}T_{5}^{j}T_{6}^{k}T_{7}^{\infty }:i,j,k\in \omega \}$.

\begin{theorem}
\label{T2.1}For any sm-pair $(U,\psi )$, the relation $dtyp(U,\psi )\in
\{T_{1}^{\infty }\}\cup \Delta $ holds. For any sequence $\sigma \in
\{T_{1}^{\infty }\}\cup \Delta $, there exists a sm-pair $(U,\psi )$ such that $dtyp(U,\psi )=\sigma $.
\end{theorem}

\begin{theorem}
\label{T2.2}For any limited sm-pair $(U,\psi )$, the relation $dtyp(U,\psi
)\in \Delta $ holds. For any sequence $\sigma \in \Delta $, there exists a
limited sm-pair $(U,\psi )$ such that $\psi $ is a
weighted depth and $dtyp(U,\psi )=\sigma $.
\end{theorem}

\section{Possible Upper $n$-Types of SM-Pairs}
\label{S4}

Let $(U,\psi )$ be a sm-pair and $n\in \omega \setminus \{0\}$. We denote by
$typ_{u}(U,\psi ,n)$ a table with three rows and three columns in which rows
from top to bottom and columns from the left to the right are labeled with
indices $i,d,a$ and the value $typ(\mathcal{U}_{U\psi n}^{bc})$ is in the
intersection of the row with index $b\in \{i,d,a\}$ and the column with
index $c\in \{i,d,a\}$. The table $typ_{u}(U,\psi ,n)$ will be called the upper $%
n $-type of sm-pair $(U,\psi )$. In this section, all possible upper $n$%
-types of sm-pairs are enumerated.

We now define seven tables:
$$
t_{1}=%
\begin{tabular}{|c|ccc|}
\hline
& $i$ & $d$ & $a$ \\ \hline
$i$ & $\alpha $ & $\alpha $ & $\alpha $ \\
$d$ & $\alpha $ & $\alpha $ & $\alpha $ \\
$a$ & $\alpha $ & $\alpha $ & $\alpha $ \\ \hline
\end{tabular}%
\;\;t_{2}=%
\begin{tabular}{|c|ccc|}
\hline
& $i$ & $d$ & $a$ \\ \hline
$i$ & $\gamma $ & $\varepsilon $ & $\varepsilon $ \\
$d$ & $\alpha $ & $\alpha $ & $\alpha $ \\
$a$ & $\alpha $ & $\alpha $ & $\alpha $ \\ \hline
\end{tabular}%
\;\;t_{3}=%
\begin{tabular}{|c|ccc|}
\hline
& $i$ & $d$ & $a$ \\ \hline
$i$ & $\gamma $ & $\varepsilon $ & $\varepsilon $ \\
$d$ & $\beta $ & $\gamma $ & $\varepsilon $ \\
$a$ & $\alpha $ & $\alpha $ & $\alpha $ \\ \hline
\end{tabular}%
$$
$$
t_{4}=%
\begin{tabular}{|c|ccc|}
\hline
& $i$ & $d$ & $a$ \\ \hline
$i$ & $\gamma $ & $\varepsilon $ & $\varepsilon $ \\
$d$ & $\gamma $ & $\gamma $ & $\varepsilon $ \\
$a$ & $\alpha $ & $\alpha $ & $\alpha $ \\ \hline
\end{tabular}%
\;\;
t_{5}=%
\begin{tabular}{|c|ccc|}
\hline
& $i$ & $d$ & $a$ \\ \hline
$i$ & $\gamma $ & $\varepsilon $ & $\varepsilon $ \\
$d$ & $\gamma $ & $\gamma $ & $\gamma $ \\
$a$ & $\gamma $ & $\gamma $ & $\gamma $ \\ \hline
\end{tabular}%
\;\;t_{6}=%
\begin{tabular}{|c|ccc|}
\hline
& $i$ & $d$ & $a$ \\ \hline
$i$ & $\gamma $ & $\varepsilon $ & $\varepsilon $ \\
$d$ & $\gamma $ & $\gamma $ & $\delta $ \\
$a$ & $\gamma $ & $\gamma $ & $\gamma $ \\ \hline
\end{tabular}%
$$
$$
t_{7}=%
\begin{tabular}{|c|ccc|}
\hline
& $i$ & $d$ & $a$ \\ \hline
$i$ & $\gamma $ & $\varepsilon $ & $\varepsilon $ \\
$d$ & $\gamma $ & $\gamma $ & $\varepsilon $ \\
$a$ & $\gamma $ & $\gamma $ & $\gamma $ \\ \hline
\end{tabular}%
$$

In this section, we will prove the following two propositions.

\begin{proposition}
\label{P3.1}For any sm-pair $(U,\psi )$ and any $n\in \omega \setminus \{0\}$%
, the relation $typ_{u}(U,\psi ,n)\in \{t_{1},\ldots ,t_{7}\}$ holds.
\end{proposition}

\begin{proposition}
\label{P3.2}For any limited sm-pair $(U,\psi )$ and any $n\in \omega
\setminus \{0\}$, the relation $typ_{u}(U,\psi ,n)\in \{t_{2},\ldots
,t_{7}\} $ holds.
\end{proposition}

First, we prove some auxiliary statements.

\begin{lemma}
\label{L3.1} Let $(U,\psi )$ be a sm-pair and $z\in \mathcal{P}(U)$. Then
the inequalities $\psi _{U}^{a}(z)\leq \psi _{U}^{d}(z)\leq \psi _{U}^{i}(z)$
hold.
\end{lemma}

\begin{proof}
Let $z=(Y,\nu ,\beta _{1},\ldots ,\beta _{m})$. It is not difficult to
construct a computation tree $\Gamma _{0}\in \mathcal{T}(U)$, which solves the
problem $z$ deterministically and for which $\beta (\xi )=\beta (z)=$ $\beta
_{1},\ldots , \beta _{m}$ for any complete path $\xi $ in the computation tree $%
\Gamma _{0}$. Evidently, $\psi (\Gamma _{0})=\psi _{U}^{i}(z)$. Therefore $%
\psi _{U}^{d}(z)\leq \psi _{U}^{i}(z)$. If a computation tree $\Gamma \in
\mathcal{T}(U)$ solves the problem $z$ deterministically, then the
computation tree $\Gamma $ solves the problem $z$ nondeterministically.
Therefore $\psi _{U}^{a}(z)\leq \psi _{U}^{d}(z)$.
\end{proof}

Let $(U,\psi )$ be a sm-pair, $n\in \omega \setminus \{0\}$, $m\in \omega $,
and $b,c\in \{i,d,a\}$. The notation $\mathcal{U}_{U\psi n}^{bc}(m)=\infty $
means that the set $\{\psi _{U}^{b}(z):z\in \mathcal{P}(U,n),\psi
_{U}^{c}(z)\leq m\}$ is infinite. Evidently, if $\mathcal{U}_{U\psi
n}^{bc}(m)=\infty $, then $\mathcal{U}_{U\psi n}^{bc}(m+1)=\infty $. It is
not difficult to prove the following statement.

\begin{lemma}
\label{L3.2} Let $(U,\psi )$ be a sm-pair, $n\in \omega \setminus \{0\}$, $%
m\in \omega $, and $b,c\in \{i,d,a\}$. Then

(a) If there exists $m\in \omega $ such that $\mathcal{U}_{U\psi
n}^{bc}(m)=\infty $, then $typ(\mathcal{U}_{U\psi n}^{bc})=\varepsilon $,
$Dom(\mathcal{U}_{U\psi n}^{bc})=\emptyset $ if $m_{0}=m_{1}$, and $Dom(%
\mathcal{U}_{U\psi n}^{bc})=\{m:m\in \omega ,m_{0}\leq m<m_{1}\}$ if $%
m_{0}<m_{1}$, where $m_{0}=\min \{\psi _{U}^{c}(z):z\in \mathcal{P}(U,n)\}$
and $m_{1}=\min \{m:m\in \omega ,\mathcal{U}_{U\psi n}^{bc}(m)=\infty \}$.

(b) If there is no $m\in \omega $ such that $\mathcal{U}_{U\psi
n}^{bc}(m)=\infty $, then $Dom(\mathcal{U}_{U\psi n}^{bc})=\{m:m\in \omega
,m\geq m_{0}\}$.
\end{lemma}

Let $(U,\psi )$ be a sm-pair, $n\in \omega \setminus \{0\}$, and $b,c,e,f\in
\{i,d,a\}$. The notation $\mathcal{U}_{U\psi n}^{bc}\vartriangleleft
\mathcal{U}_{U\psi n}^{ef}$ means that, for any $m\in \omega $, the
following conditions hold:
\begin{itemize}
\item If the value $\mathcal{U}_{U\psi n}^{bc}(m)$ is defined, then either $%
\mathcal{U}_{U\psi n}^{ef}(m)=\infty $ or the value $\mathcal{U}_{U\psi
n}^{ef}(m)$ is defined and the inequality $\mathcal{U}_{U\psi
n}^{bc}(m)\leq \mathcal{U}_{U\psi n}^{ef}(m)$ holds.

\item If $\mathcal{U}_{U\psi n}^{bc}(m)=\infty $, then $\mathcal{U}_{U\psi
n}^{ef}(m)=\infty $.
\end{itemize}

We define a linear order $\preceq $ on the set $\{\alpha ,\beta ,\gamma
,\delta ,\varepsilon \}$ as follows: $\alpha \preceq \beta \preceq \gamma
\preceq \delta \preceq \varepsilon $.

\begin{lemma}
\label{L3.3} Let $(U,\psi )$ be a sm-pair and $n\in \omega \setminus \{0\}$.
Then $typ(\mathcal{U}_{U\psi n}^{bi})\preceq typ(\mathcal{U}_{U\psi
n}^{bd})\preceq typ(\mathcal{U}_{U\psi n}^{ba})$ and $typ(\mathcal{U}_{U\psi
n}^{ab})\preceq typ(\mathcal{U}_{U\psi n}^{db})\preceq typ(\mathcal{U}%
_{U\psi n}^{ib})$ for any $b\in \{i,d,a\}$.
\end{lemma}

\begin{proof}
From the definition of the functions $\mathcal{U}_{U\psi n}^{bc}$, $b,c\in
\{i,d,a\}$, and from Lemma \ref{L3.1} it follows that $\mathcal{U}_{U\psi
n}^{bi}\vartriangleleft \mathcal{U}_{U\psi n}^{bd}\vartriangleleft \mathcal{U%
}_{U\psi n}^{ba}$ and $\mathcal{U}_{U\psi n}^{ab}\vartriangleleft \mathcal{U}%
_{U\psi n}^{db}\vartriangleleft \mathcal{U}_{U\psi n}^{ib}$ for any $b\in
\{i,d,a\}$. Using these relations and Lemma \ref{L3.2} we obtain the
statement of the lemma.
 \end{proof}

\begin{lemma}
\label{L3.4} Let $(U,\psi )$ be a sm-pair, $n\in \omega \setminus \{0\}$,
and $b,c\in \{i,d,a\}$. Then

(a) $typ(\mathcal{U}_{U\psi n}^{bc})=\alpha $ if and only if the function $%
\psi _{U}^{b}$ is bounded from above on the set $\mathcal{P}(U,n)$.

(b) If the function $\psi _{U}^{b}$ is unbounded from above on the set $%
\mathcal{P}(U,n)$, then $typ(\mathcal{U}_{U\psi n}^{bb})=\gamma $.
\end{lemma}

\begin{proof}
(a) The first statement of the lemma is obvious. (b) Let the function $\psi
_{U}^{b}$ be unbounded from above on the set $\mathcal{P}(U,n)$. One can
show that in this case the equality $\mathcal{U}_{U\psi n}^{bb}(m)=m$ holds for
infinitely many $m\in \omega $. Therefore $typ(\mathcal{U}_{U\psi
n}^{bb})=\gamma $.
 \end{proof}

\begin{corollary}
\label{C3.1} Let $(U,\psi )$ be a sm-pair, $n\in \omega \setminus \{0\}$,
and $b\in \{i,d,a\}$. Then $typ(\mathcal{U}_{U\psi n}^{bb})\in \{\alpha
,\gamma \}$.
\end{corollary}

Let $U=(A,F,P)$ be a structure and $n\in \omega \setminus \{0\}$. We denote by $%
P_{n}[F]$ the set of functions from $P[F]$ with variables from the set $%
\{x_{1},\ldots ,x_{n}\}$. It is not difficult to prove the following
statement.

\begin{lemma}
\label{L0.1}Let $(U,\psi )$ be a limited sm-pair, $U=(A,F,P)$, and $n\in
\omega \setminus \{0\}$. Then the function $\psi _{U}^{d}$ is bounded from
above on the set $\mathcal{P}(U,n)$ if and only if the set $P_{n}[F]$ is
a finite set.
\end{lemma}

\begin{lemma}
\label{L3.5} Let $(U,\psi )$ be a sm-pair, $n\in \omega \setminus \{0\}$,
and $typ(\mathcal{U}_{U\psi n}^{ii})\neq \alpha $. Then $typ(\mathcal{U}%
_{U\psi n}^{id})=typ(\mathcal{U}_{U\psi n}^{ia})=\varepsilon $.
\end{lemma}

\begin{proof}
By Lemma \ref{L3.4}, the function $\psi _{U}^{i}$ is unbounded from above on
$\mathcal{P}(U,n)$. Let $r\in \omega $. Then there exists a problem $%
z=(Y,\nu ,\beta _{1},\ldots ,\beta _{m})\in \mathcal{P}(U,n)$ such that $%
\psi _{U}^{i}(z)\geq r$. Let us consider the problem $z^{\prime }=(Y,\nu
^{\prime },\beta _{1},\ldots ,\beta _{m})$, where $\nu ^{\prime }\equiv
\{0\} $. It is clear that $\psi _{U}^{i}(z^{\prime })\geq r$. Let $\Gamma
=(Y,G)$ be a computation tree which consists of the root, the terminal node
labeled with $0$ and the edge leaving the root and entering the terminal
node. One can show that the computation tree $\Gamma $ solves the problem $%
z^{\prime }$ deterministically. Therefore $\psi _{U}^{a}(z^{\prime })\leq
\psi _{U}^{d}(z^{\prime })\leq \psi (\Gamma )=\psi (\lambda )$. Taking into
account that $r$ is an arbitrary number from $\omega $, we obtain $%
\mathcal{U}_{U\psi n}^{id}(\psi (\lambda ))=\infty $ and $\mathcal{U}_{U\psi
n}^{ia}(\psi (\lambda ))=\infty $. By Lemma \ref{L3.2}, $typ(\mathcal{U}%
_{U\psi n}^{id})=typ(\mathcal{U}_{U\psi n}^{ia})=\varepsilon $.
 \end{proof}

\begin{lemma}
\label{L3.6} Let $(U,\psi )$ be a sm-pair and $n\in \omega \setminus \{0\}$.
Then $typ(\mathcal{U}_{U\psi n}^{ai})\in \{\alpha ,\gamma \}$.
\end{lemma}

\begin{proof}
Using Lemma \ref{L3.3} and Corollary \ref{C3.1} we obtain $typ(\mathcal{U}%
_{U\psi n}^{ai})\in \{\alpha ,\beta ,\gamma \}$. Assume that $typ(\mathcal{U}%
_{U\psi n}^{ai})=\beta $. Then there exists $r\in \omega \setminus \{0\}$
such that $\mathcal{U}_{U\psi n}^{ai}(m)<m$ for any $m\in \omega $, $m>r$.
We prove by induction on $m$ that, for any problem $z\in \mathcal{P}(U,n)$,
if $\psi _{U}^{i}(z)\leq m$, then $\psi _{U}^{a}(z)\leq r_{0}$, where $%
r_{0}=\max \{r,\psi (\lambda )\}$. From Lemma \ref{L3.1} it follows that
under the condition $m\leq r$ the considered statement holds. Let it hold
for some $m\geq r.$ We now show that this statement holds for $m+1$ too. Let
$z\in \mathcal{P}(U,n)$ and $\psi _{U}^{i}(z)\leq m+1$. Since $m+1>r$, $\psi
_{U}^{a}(z)\leq m$. Let $\Gamma \in \mathcal{T}(U)$, $\psi (\Gamma )=$ $\psi
_{U}^{a}(z)$, and the computation tree $\Gamma $ solves the problem $z$
nondeterministically. Assume that in $\Gamma $ there exists a complete path $%
\xi $ such that in the sequence $\beta (\xi )$ there are no predicate
expressions. In this case, a computation tree, which has the same set of input
variables as $z$ and consists of the root, the terminal node labeled with $%
\kappa (\xi )$, and the edge leaving the root and entering the terminal
node, solves the problem $z$ nondeterministically. Therefore $\psi
_{U}^{a}(z)\leq \psi (\lambda )\leq r_{0}$. Assume now that, for each
complete path $\xi $ in the computation tree $\Gamma $, the sequence $\beta (\xi
)$ contains a predicate expression. Let $\xi \in \Xi (\Gamma )$, $\xi
=v_{0},d_{0},\ldots ,v_{p},d_{p},v_{p+1}$ and let the expression $\beta _{i}$
be assigned to the node $v_{i}$, $i=1,\ldots ,p$. Let there be exactly $t>0$
predicate expressions $\beta _{c_{1}},\ldots ,\beta _{c_{t}}$ among the
expressions $\beta _{1},\ldots ,\beta _{p}$, where $c_{1}<\cdots <c_{t}$.
For $i=1,\ldots ,t$, let $\delta _{i}$ be the number from the set $E_{2}$
assigned to the edge $d_{c_{i}}$. Let us consider a problem $z_{\xi }=(Y,\nu
_{\xi },\beta _{1},\ldots ,\beta _{p})$, where $Y$ is the set of input
variables for the problem $z$, $\nu _{\xi }(\delta _{1},\ldots ,\delta
_{t})=\{\kappa (\xi )\}$ and $\nu _{\xi }(\bar{\sigma})=\{\kappa (\xi )+1\}$
for any $t$-tuple $\bar{\sigma}\in E_{2}^{t}$ such that $\bar{\sigma}\neq $ $%
(\delta _{1},\ldots ,\delta _{t})$. It is clear that $\psi _{U}^{i}(z_{\xi
})\leq m$. Using the inductive hypothesis we obtain that there exists a
computation tree $\Gamma _{\xi }\in \mathcal{T}(U)$, which has the following
properties: $\Gamma _{\xi }$ solves the problem $z_{\xi }$
nondeterministically and $\psi (\Gamma _{\xi })\leq r_{0}$. Let $\mathcal{A}%
(\xi )\neq \emptyset $. We denote by $\tilde{\Gamma}_{\xi }$ a computation tree
obtained from $\Gamma _{\xi }$ by removal of all nodes and edges satisfying
the following condition: there is no a complete path $\xi ^{\prime }$ in $%
\Gamma _{\xi }$, which contains this node or edge and for which $\kappa (\xi
^{\prime })=\kappa (\xi )$. Let $\{\xi :\xi \in \Xi (\Gamma ),\mathcal{A}%
(\xi )\neq \emptyset \}=\{\xi _{1},\ldots ,\xi _{q}\}$. We identify the
roots of computation trees $\tilde{\Gamma}_{\xi _{1}},\ldots ,\tilde{\Gamma}%
_{\xi _{q}}$. Denote by $\Gamma ^{\ast }$ the obtained computation tree. It is
not difficult to show that $\Gamma ^{\ast }\in \mathcal{T}(U)$, $\psi
(\Gamma ^{\ast })\leq r_{0}$ and the computation tree $\Gamma ^{\ast }$ solves
the problem $z$ nondeterministically. Thus, the considered statement holds.
By Lemma \ref{L3.4}, $typ(\mathcal{U}_{U\psi n}^{ai})=\alpha $. The obtained
contradiction shows that $typ(\mathcal{U}_{U\psi n}^{ai})\in \{\alpha
,\gamma \}$.
 \end{proof}

\begin{proof}[Proof of Proposition \ref{P3.1}] Let $(U,\psi )$ be a sm-pair and $n\in \omega
\setminus \{0\}$. By Corollary \ref{C3.1}, $typ(\mathcal{U}_{U\psi
n}^{ii})\in \{\alpha ,\gamma \}$. Using Corollary \ref{C3.1} and Lemma \ref%
{L3.3} we obtain $typ(\mathcal{U}_{U\psi n}^{di})\in \{\alpha ,\beta ,\gamma
\}$. From Lemma \ref{L3.6} it follows that $typ(\mathcal{U}_{U\psi
n}^{ai})\in \{\alpha ,\gamma \}$.

(a) Let $typ(\mathcal{U}_{U\psi n}^{ii})=\alpha $. Using Lemmas \ref{L3.3}
and \ref{L3.4} we obtain $typ_{u}(U,\psi ,n)=t_{1}$.

(b) Let $typ(\mathcal{U}_{U\psi n}^{ii})=\gamma $ and $typ(\mathcal{U}%
_{U\psi n}^{di})=\alpha $. Using Lemmas \ref{L3.3}, \ref{L3.4}, and \ref%
{L3.5} we obtain $typ_{u}(U,\psi ,n)=t_{2}$.

(c) Let $typ(\mathcal{U}_{U\psi n}^{ii})=\gamma $ and $typ(\mathcal{U}%
_{U\psi n}^{di})=\beta $. Using Lemma \ref{L3.5} we obtain $typ(\mathcal{U}%
_{U\psi n}^{id})=typ(\mathcal{U}_{U\psi n}^{ia})=\varepsilon $. By Lemmas %
\ref{L3.3} and \ref{L3.6}, $typ(\mathcal{U}_{U\psi n}^{ai})=\alpha $. From
this equality and from Lemma \ref{L3.4} it follows that $typ(\mathcal{U}%
_{U\psi n}^{ad})=typ(\mathcal{U}_{U\psi n}^{aa})=\alpha $. Using the
equality $typ(\mathcal{U}_{U\psi n}^{di})=\beta $, Lemma \ref{L3.3}, and
Corollary \ref{C3.1} we obtain $typ(\mathcal{U}_{U\psi n}^{dd})=\gamma $.
From the equalities $typ(\mathcal{U}_{U\psi n}^{dd})=\gamma $, $typ(\mathcal{%
U}_{U\psi n}^{aa})=\alpha $ and from Lemmas \ref{L3.2} and \ref{L3.4} it
follows that $typ(\mathcal{U}_{U\psi n}^{da})=\varepsilon $. Thus, $%
typ_{u}(U,\psi ,n)=t_{3}$.

(d) Let $typ(\mathcal{U}_{U\psi n}^{ii})=typ(\mathcal{U}_{U\psi
n}^{di})=\gamma $ and $typ(\mathcal{U}_{U\psi n}^{ai})=\alpha $. Using Lemma %
\ref{L3.5} we obtain $typ(\mathcal{U}_{U\psi n}^{id})=typ(\mathcal{U}_{U\psi
n}^{ia})=\varepsilon $. From Lemma \ref{L3.4} it follows that $typ(\mathcal{U%
}_{U\psi n}^{ad})=typ(\mathcal{U}_{U\psi n}^{aa})=\alpha $. Using Lemma \ref%
{L3.3} and Corollary \ref{C3.1} we obtain $typ(\mathcal{U}_{U\psi
n}^{dd})=\gamma $. Taking into account the equality $typ(\mathcal{U}_{U\psi
n}^{aa})=\alpha $ and Lemmas \ref{L3.2} and \ref{L3.4} we obtain $typ(%
\mathcal{U}_{U\psi n}^{da})=\varepsilon $. Thus, $typ_{u}(U,\psi ,n)=t_{4}$.

(e) Let $typ(\mathcal{U}_{U\psi n}^{ii})=typ(\mathcal{U}_{U\psi n}^{di})=typ(%
\mathcal{U}_{U\psi n}^{ai})=\gamma $. Using Lemma \ref{L3.5} we obtain $typ(%
\mathcal{U}_{U\psi n}^{id})=typ(\mathcal{U}_{U\psi n}^{ia})=\varepsilon $.
From Lemma \ref{L3.3} and Corollary \ref{C3.1} it follows that $typ(\mathcal{%
U}_{U\psi n}^{dd})=typ(\mathcal{U}_{U\psi n}^{ad})=typ(\mathcal{U}_{U\psi
n}^{aa})=\gamma $. By Lemma \ref{L3.3}, $typ(\mathcal{U}_{U\psi n}^{da})\in
\{\gamma ,\delta ,\varepsilon \}$. Therefore $typ_{u}(U,\psi ,n)\in
\{t_{5},t_{6},t_{7}\}$.
 \end{proof}

\begin{proof}
[Proof of Proposition \ref{P3.2}] Let $(U,\psi )$ be a limited sm-pair, $U=(A,F,P)$%
, and $n\in \omega \setminus \{0\}$. Then, for any $\alpha \in (F\cup
P)^{\ast }$, the inequality $\psi (\alpha )\geq |\alpha |$ holds. Therefore the function $\psi _{U}^{i}$ is unbounded from above on the set $\mathcal{P}(U,n)$.
By Lemma %
\ref{L3.4}, $typ(\mathcal{U}_{U\psi n}^{ii})\neq \alpha $. Therefore $%
typ_{u}(U,\psi ,n)\neq t_{1}$. From this relation and Proposition \ref{P3.1}
it follows that the statement of Proposition \ref{P3.2} holds.
 \end{proof}

\section{Possible Upper Dynamic Types of SM-Pairs}
\label{S5}

Let $(U,\psi )$ be a sm-pair. The infinite sequence $typ_{u}(U,\psi
,1)typ_{u}(U,\psi ,2)\cdots $ is called the upper dynamic type of the sm-pair $%
(U,\psi )$. We denote this sequence $dtyp_{u}(U,\psi )$. Denote $\Delta
_{u}=\{t_{2}^{\infty },t_{2}^{i}t_{3}^{\infty
},t_{2}^{i}t_{3}^{j}t_{4}^{\infty },t_{2}^{i}t_{3}^{j}t_{4}^{k}t_{7}^{\infty
},t_{2}^{i}t_{5}^{\infty },t_{2}^{i}t_{5}^{j}t_{6}^{\infty
},t_{2}^{i}t_{5}^{j}t_{6}^{k}t_{7}^{\infty }:i,j,k\in \omega \}$. In this
section, we prove the following two propositions.

\begin{proposition}
\label{P4.1}For any sm-pair $(U,\psi )$, the relation $dtyp_{u}(U,\psi )\in
\{t_{1}^{\infty }\}\cup \Delta _{u}$ holds.
\end{proposition}

\begin{proposition}
\label{P4.2}For any limited sm-pair $(U,\psi )$, the relation $%
dtyp_{u}(U,\psi )\in \Delta _{u}$ holds.
\end{proposition}

Let $(U,\psi )$ be a sm-pair, $n\in \omega \setminus \{0\}$, and $b,c\in
\{i,d,a\}$. The notation $\mathcal{U}_{U\psi n}^{bc}\vartriangleleft
\mathcal{U}_{U\psi n+1}^{bc}$ means that, for any $m\in \omega $, the
following conditions hold:
\begin{itemize}
\item If the value $\mathcal{U}_{U\psi n}^{bc}(m)$ is defined, then either $%
\mathcal{U}_{U\psi n+1}^{bc}(m)=\infty $ or the value $\mathcal{U}_{U\psi
n+1}^{bc}(m)$ is defined and the inequality $\mathcal{U}_{U\psi
n}^{bc}(m)\leq \mathcal{U}_{U\psi n+1}^{bc}(m)$ holds.

\item If $\mathcal{U}_{U\psi n}^{bc}(m)=\infty $, then $\mathcal{U}_{U\psi
n+1}^{bc}(m)=\infty $.
\end{itemize}

\begin{lemma}
\label{L4.1} Let $(U,\psi )$ be a sm-pair, $n\in \omega \setminus \{0\}$,
and $b,c\in \{i,d,a\}$. Then $typ(\mathcal{U}_{U\psi n}^{bc})\preceq typ(%
\mathcal{U}_{U\psi n+1}^{bc})$.
\end{lemma}

\begin{proof}
Let $U=(A,F,P)$. Let $z\in \mathcal{P}(U,n)$, $z=(Y,\nu ,\beta _{1},\ldots
,\beta _{m})$, and $s$ be the minimum number from $\omega $ such that all
variables from $Y$ and all variables from the expressions $\beta _{1},\ldots
,\beta _{m}$ belong to the set $\{x_{0},\ldots ,x_{s}\}$. Denote $r=s+1$ and
$\tilde{z}=(Y\cup \{x_{r}\},\nu ,\beta _{1},\ldots ,\beta _{m})$. We will prove
that, for any $b\in \{i,d,a\}$, the equality $\psi _{U}^{b}(z)=\psi _{U}^{b}(%
\tilde{z})$ holds. It is clear that the considered equality holds if $b=i$.
We now consider the case $b=a$.

One can show that there exists a computation tree $\Gamma _{1}=(Y,G_{1})$ over $%
U $ satisfying the following conditions:
\begin{itemize}
\item $\Gamma _{1}$ solves the problem $z$ nondeterministically.
\item $\psi (\Gamma _{1})=\psi _{U}^{a}(z)$.

\item The variable $x_{r}$ is not contained in the expressions assigned to
nodes of $\Gamma _{1}$.
\end{itemize}

Denote $\Gamma _{2}=(Y\cup \{x_{r}\},G_{1})$. Let $\bar{a}^{\prime
}=(a_{1},\ldots ,a_{n},a_{n+1})\in A^{n+1},$ $\bar{a}=(a_{1},\ldots ,a_{n})$%
, $\xi ^{\prime }$ be a complete path in $\Gamma _{2}$, and $\xi $ be the
complete path in $\Gamma _{1}$, which coincides with $\xi ^{\prime }$. One
can show that $\tilde{z}(\bar{a}^{\prime })=z(\bar{a})$ and $\bar{a}^{\prime
}\in \mathcal{A}(\xi ^{\prime })$ if and only if $\bar{a}\in \mathcal{A}(\xi
)$. Using these relations it is not difficult to show that $\Gamma _{2}$
solves the problem $\tilde{z}$ nondeterministically. Therefore $\psi
_{U}^{a}(\tilde{z})\leq \psi _{U}^{a}(z)$.

Let $\Gamma _{3}=(Y\cup \{x_{r}\},G_{3})$ be a computation tree over $U$, which
solves the problem $\tilde{z}$ nondeterministically and for which $\psi
(\Gamma _{3})=$ $\psi _{U}^{a}(\tilde{z})$. Denote $\Gamma _{4}=(Y,G_{3})$.
Let $\bar{a}=(a_{1},\ldots ,a_{n})\in A^{n}$, $\bar{a}^{\prime
}=(a_{1},\ldots ,a_{n},a_{1})$, $\xi $ be a complete path in $\Gamma _{4}$
and $\xi ^{\prime }$ be the complete path in $\Gamma _{3}$,
which coincides with $\xi $. One can show that $z(\bar{a})=\tilde{z}(\bar{a}%
^{\prime })$ and $\bar{a}\in \mathcal{A}(\xi )$ if and only if $\bar{a}%
^{\prime }\in \mathcal{A}(\xi ^{\prime })$. Using these relations one can
show that $\Gamma _{4}$ solves the problem $z$ nondeterministically.
Therefore $\psi _{U}^{a}(z)\leq \psi _{U}^{a}(\tilde{z})$. Hence $\psi
_{U}^{a}(z)=\psi _{U}^{a}(\tilde{z})$.

The case $b=d$ can be considered in a similar way. Thus, $\psi
_{U}^{b}(z)=\psi _{U}^{b}(\tilde{z})$ for any problem $z\in \mathcal{P}(U,n)$
and for any $b\in \{i,d,a\}$. Taking into account that $z$ is an arbitrary
problem from $\mathcal{P}(U,n)$, one can show that $\mathcal{U}_{U\psi
n}^{bc}\vartriangleleft \mathcal{U}_{U\psi n+1}^{bc}$. Using this relation
and Lemma \ref{L3.2} we obtain $typ(\mathcal{U}_{U\psi n}^{bc})\preceq typ(%
\mathcal{U}_{U\psi n+1}^{bc})$.
 \end{proof}

We now define a partial order $\leqq $ on the set of tables $\{t_{1},\ldots
,t_{7}\}$. For $b,c\in \{i,d,a\}$ and $i\in \{1,\ldots ,7\}$, let $%
t_{i}^{bc} $ be the letter from the set $\{\alpha ,\beta ,\gamma \,,\delta
,\,\varepsilon \}$ that is in the table $t_{i}$ in the intersection of the
row with index $b$ and the column with index $c$. Let $i,j\in \{1,\ldots
,7\} $. Then $t_{i}\leqq t_{j}$ if and only if $t_{i}^{bc}\preceq t_{j}^{bc}$
for any $b,c\in \{i,d,a\}$. It is easy to check that the graph depicted in
Fig. \ref{Fig1} is the Hasse diagram for the partially ordered set $(\{t_{1},\ldots
,t_{7}\},\leqq )$. Nodes of this diagram are tables $t_{1},\ldots ,t_{7}$.
An edge goes upward from $t_{i}$ to $t_{j}$ if $t_{i}\leqq t_{j}$ and there
is no $t_{k}$, $t_{k}\notin \{t_{i},t_{j}\}$ such that $t_{i}\leqq
t_{k}\leqq t_{j}$.

\begin{figure}[th]
\centering
\includegraphics[width=0.5\textwidth]{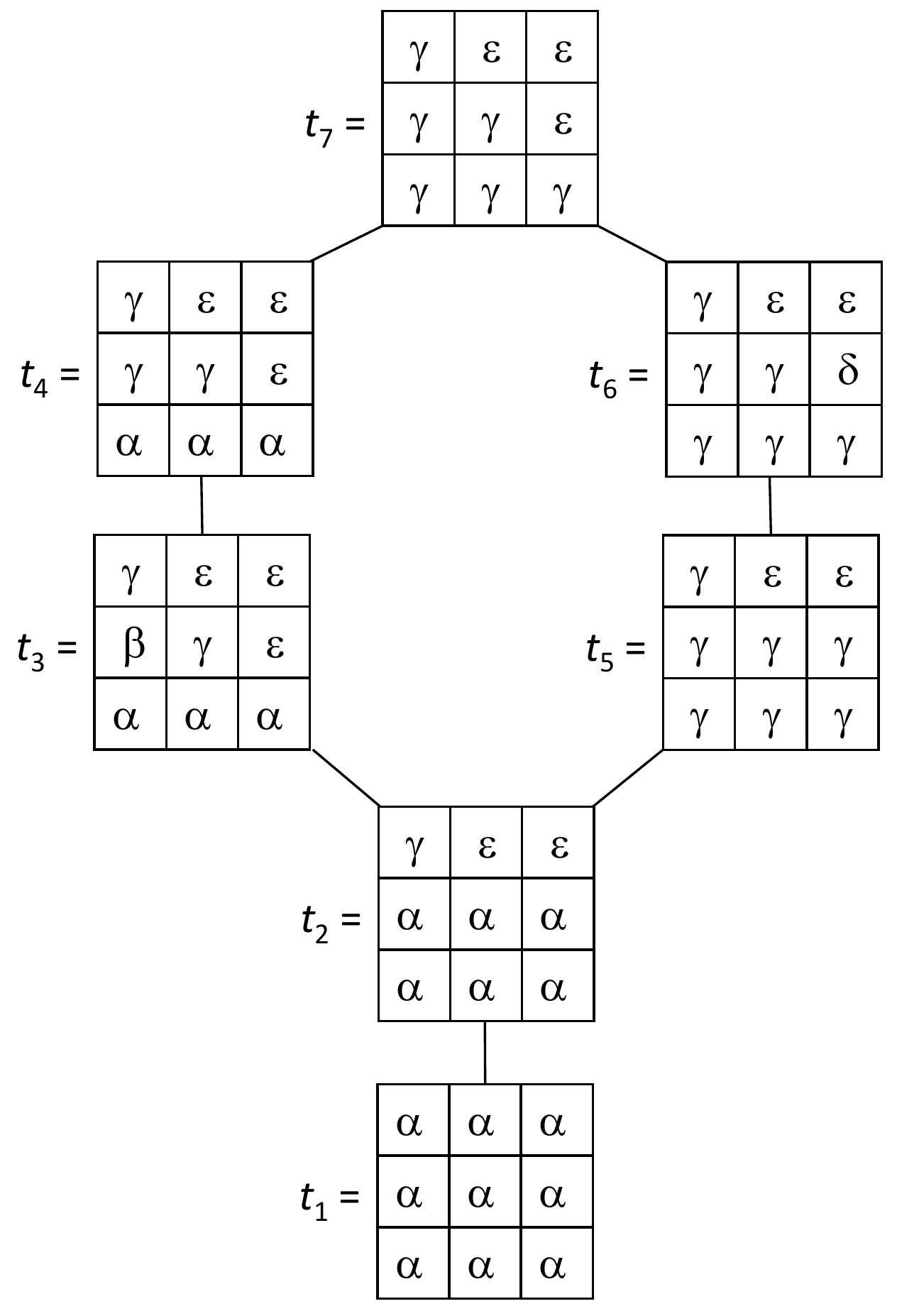}
\caption{Hasse diagram for the partially ordered set $(\{t_{1},\ldots
,t_{7}\},\leqq )$}
\label{Fig1}
\end{figure}

\begin{proof}
[Proof of Proposition \ref{P4.1}] Let $(U,\psi )$ be a sm-pair. Assume that there
exists $n_{0}\in \omega \setminus \{0\}$ for which $typ_{u}(U,\psi
,n_{0})=t_{1}$. Then, by Lemma \ref{L3.4}, the function $\psi _{U}^{i}$ is
bounded from above on the set $\mathcal{P}(U,n_{0})$. From here it follows
that the function $\psi _{U}^{i}$ is bounded from above on the set $\mathcal{%
P}(U,n)$ for any $n\in \omega \setminus \{0\}$. Using Lemmas \ref{L3.1} and %
\ref{L3.4} we obtain $dtyp_{u}(U,\psi )=t_{1}^{\infty }$.

Let for any $n\in \omega \setminus \{0\}$, the relation $typ_{u}(U,\psi
,n)\neq t_{1}$ hold. In this case, by Proposition \ref{P3.1}, $typ_{u}(U,\psi ,n)\in
\{t_{2},\ldots ,t_{7}\}$ for any $n\in \omega \setminus \{0\}$. Using Lemma %
\ref{L4.1}, we obtain that $typ_{u}(U,\psi ,n)\leqq typ_{u}(U,\psi ,n+1)$
for any $n\in \omega \setminus \{0\}$. Simple analysis of the Hasse diagram
for the partially ordered set $(\{t_{1},\ldots ,t_{7}\},\leqq )$ depicted in
Fig. 1 shows that  the set $\Delta _{u}$ coincides with the set of
infinite sequences $t_{i_{1}}t_{i_{2}}\cdots $ such that
$t_{i_{1}},t_{i_{2}},\ldots \in \{t_{2},\ldots ,t_{7}\}$ and $%
t_{i_{1}}\leqq t_{i_{2}}\leqq \cdots $.
Therefore $dtyp_{u}(U,\psi )\in \Delta _{u}$
 \end{proof}

\begin{proof}
[Proof of Proposition \ref{P4.2}] Let $(U,\psi )$ be a limited sm-pair. Using
Proposition \ref{P3.2} we obtain $dtyp_{u}(U,\psi )\neq t_{1}^{\infty }$.
From this relation and Proposition \ref{P4.1} it follows that $%
dtyp_{u}(U,\psi )\in \Delta _{u}$.
 \end{proof}

\section{Realizable Upper Dynamic Types of SM-Pairs}
\label{S6}

In this section, we prove the following two propositions.

\begin{proposition}
\label{P5.1}For any sequence $\tau \in \{t_{1}^{\infty }\}\cup \Delta _{u}$,
there exists a sm-pair $(U,\psi )$ such that $dtyp_{u}(U,\psi )=\tau $.
\end{proposition}

\begin{proposition}
\label{P5.2}For any sequence $\tau \in \Delta _{u}$, there exists a limited
sm-pair $(U,\psi )$ such that $dtyp_{u}(U,\psi )=\tau $ and $\psi $ is a
weighted depth.
\end{proposition}

We now describe a construction that will be used in the proofs of the
considered propositions.

For $i=2,\ldots ,7$, we define a sm-pair $\pi _{i}=(U_{i},\psi _{i})$, where
$U_{i}=(A_{i},F_{i},P_{i})$, $F_{i}=\emptyset $, $P_{i}$ is a set of
one-place predicates, and $\psi _{i}$ is a weighted depth. It is clear that
for the definition of the function $\psi _{i}$ it is enough to define values
of $\psi _{i}$ on elements of the set $P_{i}$.

Define the sm-pair $\pi _{2}$ as follows: $A_{2}=\{0\}$, $P_{2}=\{q_1\}$, $%
q_1(0)=0$, and $\psi _{2}(q_1)=1$.

Define the sm-pair $\pi _{3}$ as follows: $A_{3}=\omega $, $%
P_{3}=\{l_{i}:i\in \omega \}$,
\[
l_{i}(j)=\left\{
\begin{array}{cc}
0, & j\leq i, \\
1, & j>i,%
\end{array}%
\right.
\]%
and $\psi _{3}(l_{i})=1$ for any $i,j\in \omega $, i.e., $\psi _{3}=h$.

Define the sm-pair $\pi _{4}$ as follows: $A_{4}=\omega $, $P_{4}\ $is the
set of mappings of the kind $f:\omega \rightarrow \{0,1\}$, and $\psi
_{4}(f)=1$ for any $f\in P_{4}$, i.e., $\psi _{4}=h$.

Define the sm-pair $\pi _{5}$ as follows: $A_{5}=\omega $, $%
P_{5}=\{q_{i}:i\in \omega \}$,
\[
q_{i}(j)=\left\{
\begin{array}{cc}
0, & j\neq i, \\
1, & j=i%
\end{array}%
\right.
\]%
for any $i,j\in \omega $, $\psi _{5}(q_{0})=1$, and $\psi _{5}(q_{i})=i\sum_{t=0}^{i-1}\psi
_{5}(q_{t})$ for $i\geq 1$.

Define the sm-pair $\pi _{6}$ as follows: $A_{6}=\omega $, $%
P_{6}=\{q_{2i},q_{2i+1},p_{2i}:i\in \omega \backslash \{0\}\}$, $q_{2i}(j)=1$
if and only if $j=2i$, $q_{2i+1}(j)=1$ if and only if $j=2i+1$, $p_{2i}(j)=1$
if and only if $j=2i$ or $j=2i+1$, and $\psi _{6}(q_{2i})=\psi
_{6}(q_{2i+1})=\psi _{6}(p_{2i})=i$ for any $i\in \omega \setminus \{0\}$
and $j\in \omega $.

Define the sm-pair $\pi _{7}$ as follows: $A_{7}=\mathbb{Z}$, where $\mathbb{Z}$ is the set of integers, $P_{7}=\{l_{i}:i\in \omega \}\cup \{q_{-k}:k\in
\omega \setminus \{0\}\}$,
\[
l_{i}(j)=\left\{
\begin{array}{cc}
0, & j\leq i, \\
1, & j>i,%
\end{array}%
\right.
\;\;
q_{-k}(j)=\left\{
\begin{array}{cc}
0, & j\neq -k, \\
1, & j=-k,%
\end{array}%
\right.
\]%
$\psi _{7}(l_{i})=1$, and $\psi _{7}(q_{-k})=k$ for any $i\in \omega $, $%
k\in \omega \setminus \{0\}$, and $j\in \mathbb{Z}$.

Let $r\in \{2,\ldots ,7\}$ and $n\in \omega \setminus \{0\}$. Define a
sm-pair $\pi _{r}^{(n)}=(U_{r}^{(n)},\psi _{r}^{(n)})$, where $%
U_{r}^{(n)}=(A_{r}^{(n)},F_{r}^{(n)},P_{r}^{(n)})$ and $\psi _{r}^{(n)}$ is a weighted
depth. Let $K=\{k_{i}:i\in \omega \}$ be a set such that $k_i \neq k_j$ if $i\neq j$ and $K\cap
\mathbb{Z}
=\emptyset $. For any $c\in A_{r}\cup K$, denote $c^{(n)}=(c,n)$. Then $%
A_{r}^{(n)}=\{a^{(n)}:a\in A_{r}\}\cup \{k_{0}^{(n)},k_{1}^{(n)},\ldots
,k_{n-1}^{(n)}\}$, $F_{r}^{(n)}=\emptyset $, $P_{r}^{(n)}=\{g^{(n)}:g\in
P_{r}\}$, where
\begin{equation}
g^{(n)}(x_{1},\ldots ,x_{n})=\left\{
\begin{array}{rl}
g(a), & x_{1}=k_{1}^{(n)},\ldots ,x_{n-1}=k_{n-1}^{(n)}, \\ & x_{n}=(a,n)\in
A_{r}^{(n)}, \\
0, & \mathrm{otherwise,}%
\end{array}%
\right.  \label{E1}
\end{equation}%
and $\psi _{r}^{(n)}(g^{(n)})=\psi _{r}(g)$ for any $g\in P_{r}$.

Let $i_{1},\ldots ,i_{m}\in \{2,\ldots ,7\}$, $n_{1},\ldots ,n_{m}\in \omega
\setminus \{0\}$, and $n_{1}<\cdots <n_{m}$. Define the sm-pair $\pi
_{i_{1}}^{(n_{1})}\oplus \cdots \oplus \pi _{i_{m}}^{(n_{m})}$ as follows: $%
\pi _{i_{1}}^{(n_{1})}\oplus \cdots \oplus \pi _{i_{m}}^{(n_{m})}=(U,\psi ),$
$U=(A,F,P)$, $A=A_{i_{1}}^{(n_{1})}\cup \cdots \cup
A_{i_{m}}^{(n_{m})}$, $F=\emptyset $, $P=P_{i_{1}}^{(n_{1})}\cup \cdots \cup
P_{i_{m}}^{(n_{m})}$, and $\psi $ is a weighted depth. For any $j\in
\{1,\ldots ,m\}$ and $g^{(n_{j})}\in P_{i_{j}}^{(n_{j})}$, $g^{(n_{j})}$ is
equal to $0$ on tuples that do not belong to the set $%
(A_{i_{j}}^{(n_{j})})^{n_{j}}$, and is defined by (\ref{E1})
with $n=n_j$ and $r=i_j$
on tuples from $%
(A_{i_{j}}^{(n_{j})})^{n_{j}}$, and $\psi (g^{(n_{j})})=\psi
_{i_{j}}(g^{(n_{j})})$.

Let $\tau \in \Delta _{u}$ and $\tau =t_{v_{1}}^{w_{1}}\cdots
t_{v_{m}}^{w_{m}}$, where $v_{1},\ldots ,v_{m}$ are pairwise different
numbers from $\{2,\ldots ,7\}$, $w_{m}=\infty $ and if $m\geq 2$, then $%
w_{j} \in \omega \setminus \{0\}$ for $j=1,\ldots ,m-1$. Define the sm-pair $(U_{\tau },\psi _{\tau
}) $ as follows: $$(U_{\tau },\psi _{\tau })=\pi _{v_{1}}^{(1)}\oplus \pi
_{v_{2}}^{(w_{1}+1)}\oplus \cdots \oplus \pi _{v_m}^{(w_{1}+\cdots
+w_{m-1}+1)}. $$ In particular, if $m=1$, then $(U_{\tau },\psi _{\tau })=\pi _{v_{1}}^{(1)}$.

Let $\tau \in \Delta _{u}$, $(U_{\tau },\psi _{\tau })=\pi
_{v_{1}}^{(n_{1})}\oplus \cdots \oplus \pi _{v_{m}}^{(n_{m})}$, $n\in \omega
\setminus \{0\}$, $r\in \{1,\ldots ,m\}$, $n\geq n_{r}$ and if $r<m$, then $%
n<n_{r+1}$. Let $U_{\tau }=(A,F,P)$. It is clear that all functions from $%
P[F]$ are of the kind $g(x_{i_{1}},\ldots ,x_{i_{s}})$, where $g\in P$.
Since $n<n_{r+1}$, any function from $P[F]$ that depends on $n$  variables
and does not equal identically to $0$ is of the kind $g(x_{i_{1}},\ldots
,x_{i_{n_{j}}})$, where $j\leq r$, $g\in P_{v_{j}}^{(n_{j})}$ and $%
x_{i_{1}},\ldots ,x_{i_{n_{j}}}$ are pairwise different variables. The tuple
$(x_{i_{1}},\ldots ,x_{i_{n_{j}}})$ will be called the sort of the
considered function. Let $Y$ be a set of variables and $|Y|=n$. Then a
function depending on $n_{j}$ pairwise different variables from $Y$ may have
exactly $(n)_{n_{j}}=n(n-1)\cdots (n-n_{j}+1)$ different sorts.

We now prove some statements about properties of sm-pairs $(U_{\tau },\psi
_{\tau })$, where $\tau \in \Delta _{u}$.

\begin{lemma}
\label{L5.1} Let $\tau \in \Delta _{u}$, $(U_{\tau },\psi _{\tau })=\pi
_{v_{1}}^{(n_{1})}\oplus \cdots \oplus \pi _{v_{m}}^{(n_{m})}$, $r\in
\{1,\ldots ,m\}$, $v_{r}=2$, $n\in \omega \setminus \{0\}$, $n\geq n_{r}$
and if $r<m$, then $n<n_{r+1}$. Then $typ_{u}(U_{\tau },\psi _{\tau
},n)=t_{2}$.
\end{lemma}

\begin{proof}
Since $v_{r}=2,$ $r=1$. Denote $(U,\psi )=(U_{\tau },\psi _{\tau })$. It is
clear that the function $\psi _{U}^{i}$ is unbounded from above on the set $%
\mathcal{P}(U,n)$. Using Lemma \ref{L3.4} we obtain that $typ(\mathcal{U}%
_{U\psi n}^{ii})=\gamma $. Let $U=(A,F,P)$. Taking into account that in the
case $m\geq 2$ the inequality $n<n_{2}$ holds, one can show that all
functions from $P[F]$ depending on $n$ variables are equal identically to $0$%
. Therefore the function $\psi _{U}^{d}$ is bounded from above on the set $%
\mathcal{P}(U,n)$. From here and from Lemma \ref{L3.4} it follows that $typ(%
\mathcal{U}_{U\psi n}^{di})=\alpha $. Using Proposition \ref{P3.1} we obtain
$typ_{u}(U_{\tau },\psi _{\tau },n)=t_{2}$.
 \end{proof}

\begin{lemma}
\label{L5.2} Let $\tau \in \Delta _{u}$, $(U_{\tau },\psi _{\tau })=\pi
_{v_{1}}^{(n_{1})}\oplus \cdots \oplus \pi _{v_{m}}^{(n_{m})}$, $r\in
\{1,\ldots ,m\}$, $v_{r}=3$, $n\in \omega \setminus \{0\}$, $n\geq n_{r}$
and if $r<m$, then $n<n_{r+1}$. Then $typ_{u}(U_{\tau },\psi _{\tau
},n)=t_{3}$.
\end{lemma}

\begin{proof}
One can show that $r=1$ or $r=2$. In the latter case, $v_{1}=2$. Denote $%
(U,\psi )=(U_{\tau },\psi _{\tau })$. Let $U=(A,F,P)$. Since $n<n_{r+1}$,
all functions from $P[F]$, which depend on $n$ variables and do not equal
identically to $0$ are of the kind $l_{i}^{(n_{r})}(x_{j_{1}},\ldots
,x_{j_{n_{r}}})$, where $x_{j_{1}},\ldots ,x_{j_{n_{r}}}$ are pairwise
different variables.

We now prove that $typ(\mathcal{U}_{U\psi n}^{di})=\beta $. Using Corollary %
\ref{C3.1} and Lemma \ref{L3.3} we obtain $typ(\mathcal{U}_{U\psi
n}^{di})\in \{\alpha ,\beta ,\gamma \}$. Show that $typ(\mathcal{U}_{U\psi
n}^{di})\neq \gamma $. Let $m\in \omega \setminus \{0\}$, $z=(Y,\nu ,\beta
_{1},\ldots ,\beta _{t})\in \mathcal{P}(U,n)$, and $t\leq m$. One can
transform the problem $z$ into a problem $z^{\prime }=(Y,\nu ' ,\alpha
_{1},\ldots ,\alpha _{t'})$ such that $z(\bar{a})=z^{\prime }(\bar{a})$ for
any $\bar{a}\in A^{n}$, $t' \leq t$ and, for $s=1,\ldots ,t'$, the expression $\alpha _{s}$
is an expression of the kind $l_{i}^{(n_{r})}(x_{j_{1}},\ldots
,x_{j_{n_{r}}})$, where $x_{j_{1}},\ldots
,x_{j_{n_{r}}}$ are pairwise different variables from $Y$. There are
exactly $(n)_{n_{r}}$ different sorts of such functions.
Using
an approach similar to the binary search algorithm it is not difficult to show that, for any $c$ functions of the
same sort, there exists a computation tree over $U,$ which computes values of the
considered $c$ functions and which depth is at most $1+\log _{2}c$. Using
this fact it is not difficult to show that $h_{U}^{d}(z^{\prime })\leq
(n)_{n_{r}}(1+\log _{2}t')$ and $h_{U}^{d}(z)\leq (n)_{n_{r}}(1+\log _{2}t)$.
Taking into account that $z$ is an arbitrary problem from $\mathcal{P}(U,n)$
such that $\psi _{U}^{i}(z)\leq m$, we obtain $\mathcal{U}_{U\psi
n}^{di}(m)\leq $ $(n)_{n_{r}}(1+\log _{2}m)$. Therefore $\mathcal{U}_{U\psi
n}^{di}(m)<m$ for large enough $m$ and $typ(\mathcal{U}_{U\psi n}^{di})\neq
\gamma $.

We now prove that $typ(\mathcal{U}_{U\psi n}^{di})\neq \alpha $. Assume the
contrary. Using Lemma \ref{L3.4} we obtain that there exists a number $p\in
\omega $ such that $\psi _{U}^{d}(z) \le p$ for any problem $z\in \mathcal{P}(U,n)$%
. Therefore there exists a number $q\in \omega \setminus \{0\}$ that
satisfies the following condition: for any problem $z\in \mathcal{P}(U,n)$,
there exists a computation tree $\Gamma \in \mathcal{T}(U)$, which solves the
problem $z$ deterministically and has at most $q$ terminal nodes. Let us
consider a problem $$z^{\prime }=(Y,\nu ,l_{1}^{(n_{r})}(x_{1},\ldots
,x_{n_{r}}),\ldots ,l_{q}^{(n_{r})}(x_{1},\ldots ,x_{n_{r}}))$$ from $%
\mathcal{P}(U,n)$, where $Y=\{x_{1},\ldots ,x_{n}\}$, $\nu
:E_{2}^{q}\rightarrow \mathcal{S}(\omega )$ and $\nu (\bar{\delta}_{1})\cap
\nu (\bar{\delta}_{2})=\emptyset $ for any $\bar{\delta}_{1},\bar{\delta}%
_{2}\in E_{2}^{q}$ such that $\bar{\delta}_{1}\neq \bar{\delta}_{2}$. It is
easy to check that there are $q+1$ tuples $\bar{a}_{1},\ldots ,\bar{a}%
_{q+1}\in A^{n}$, such that $z(\bar{a}_{i})\cap z(\bar{a}_{j})=\emptyset $
for any $i,j\in \{1,\ldots ,q+1\}$,  $i\neq j$. Let $\Gamma $ be an
arbitrary computation tree over $U$, which solves the problem $z^{\prime }$
deterministically. It is not difficult to show that the tree $\Gamma $
should have at least $q+1$ terminal nodes. We obtain a contradiction.
Therefore $typ(\mathcal{U}_{U\psi n}^{di})=\beta $. From this equality and
Proposition \ref{P3.1} it follows that $typ_{u}(U_{\tau },\psi _{\tau
},n)=t_{3}$.
 \end{proof}

\begin{lemma}
\label{L5.3} Let $\tau \in \Delta _{u}$, $(U_{\tau },\psi _{\tau })=\pi
_{v_{1}}^{(n_{1})}\oplus \cdots \oplus \pi _{v_{m}}^{(n_{m})}$, $r\in
\{1,\ldots ,m\}$, $v_{r}=4$, $n\in \omega \setminus \{0\}$, $n\geq n_{r}$
and if $r<m$, then $n<n_{r+1}$. Then $typ_{u}(U_{\tau },\psi _{\tau
},n)=t_{4}$.
\end{lemma}

\begin{proof}
One can show that if $r>1$, then $\{v_{1},\ldots ,v_{r-1}\}\subseteq \{2,3\}$%
. Denote $(U,\psi )=(U_{\tau },\psi _{\tau })$. Let $U=(A,F,P)$.

We show that $typ(\mathcal{U}_{U\psi n}^{ai})=\alpha $. Let $Y$ be a set of $%
n$ variables. A system of equations%
\begin{equation}
\{g_{1}=\delta _{1},\ldots ,g_{m}=\delta _{m}\},  \label{E2}
\end{equation}%
where $g_{1},\ldots ,g_{m}$ are functions from $P$ depending on variables
from $Y$ and $\delta _{1},\ldots ,\delta _{m}\in E_{2}$, is called a system
of equations over $Y$. We show that, for any consistent on $A^{n}$ system
of equations over $Y$, there exists a system of equations over $Y$, which
has the same set of solutions from $A^{n}$ and contains at most $%
m_{0}=2\sum_{j=1}^{r}(n)_{n_j}$ equations. Consider a consistent on $A^{n}$
system of equations (\ref{E2}). Remove from this system all equations $%
g_{i}=\delta _{i}$ such that $g\equiv 0$. Denote the obtained system by $S$.
This system can contain functions of $(n)_{n_{r}}$ sorts from $%
P_{4}^{(n_{r})}$ and if $v_{t}=3$ for some $t\in \{1,\ldots ,r-1\}$, then
this system can contain functions of $(n)_{n_{t}}$ sorts from $%
P_{3}^{(n_{t})}$ too.

Divide the system $S$ into at most $(n)_{n_{r}}+(n)_{n_{t}}$ subsystems each
of which contains equations composed from functions of the same sort and from the same set $P_{4}^{(n_{r})}$ or $P_{3}^{(n_{t})}$. Let $%
\Sigma $ be one of such subsystems. Let $\Sigma $ consist of equations with
the left parts from $P_{4}^{(n_{r})}$. One can show that there exists a
function $g\in P_{4}^{(n_{r})}$ depending on variables from $Y$ and a number
$\delta \in E_{2}$ such that the set of solutions on $A^{n}$ of the equation
system $\Sigma $ coincides with the set of solutions on $A^{n}$ of the
equation system $\Sigma ^{\prime }=\{g=\delta \}$. Let now $\Sigma $ consist of equations with the left
parts from $P_{3}^{(n_{t})}$. One can show that there exists a subsystem $%
\Sigma ^{\prime }$ of the system $\Sigma $, which has the same set of
solutions as the system $\Sigma $ and contains at most two equations.
Replace in the system $S$ each subsystem $\Sigma $ with the corresponding subsystem $%
\Sigma ^{\prime }$ and denote the obtained system by $S^{\prime }$. It is
clear that $S^{\prime }$ has the same set of solutions as $S$ and contains at
most $m_{0}$ equations.

We now show that, for any problem $z=(Y,\nu ,\beta _{1},\ldots ,\beta
_{m})\in \mathcal{P}(U,n)$, the inequality $\psi _{U}^{a}(z)\leq m_{0}$
holds. It is clear that there exists a problem $z^{\prime }=(Y,\nu ,\alpha
_{1},\ldots ,\alpha _{m})\in \mathcal{P}(U,n)$ such that $\alpha _{1},\ldots
,\alpha _{m}$ are expressions depending on variables from $Y$ and  $z(\bar{a}%
)=z^{\prime }(\bar{a})$ for any $\bar{a}\in A^{n}$. For any $\bar{\delta}=(\delta _{1},\ldots ,\delta _{m})\in E_{2}^{m}$,
denote by $S(\bar{\delta})$ the system of equations%
\[
\{\alpha _{1}=\delta _{1},\ldots ,\alpha _{m}=\delta _{m}\}.
\]

Let the system $S(\bar{\delta})$ be consistent on the set $A^{n}$. Then
there exists a system of equations
\[
\{\gamma _{1}=\sigma _{1},\ldots ,\gamma _{t}=\sigma _{t}\}
\]%
over $Y$ such that $t\leq m_{0}$ and the set of solutions of this system
coincides with the set of solutions of the system $S(\bar{\delta}).$

Denote by $\Gamma (\bar{\delta})$ the computation tree $(Y,G(\bar{\delta}))$,
where $G(\bar{\delta})$ is a tree with the root consisting of unique
complete path $v_{0},d_{0},,\ldots ,v_{t},d_{t},v_{t+1}$ in which, for $%
i=1,\ldots ,t$, the node $v_{i}$ is labeled with the expression $\gamma _{i}$
and the edge $d_{i}$ is labeled with the number $\sigma _{i}$, and the node $%
v_{i+1}$ is labeled with the minimum number from the set $\nu (\bar{\delta})$%
.

Identify the roots of the trees $G(\bar{\delta})$, where $\bar{\delta}\in
E_{2}^{m}$ and the system of equations $S(\bar{\delta})$ is consistent.
Denote the obtained tree by $G$. By $\Gamma $ we denote the pair $(Y,G)$. It
is not difficult to prove that $\Gamma $ is a computation tree over $U$, which
solves the problem $z$ nondeterministically and for which $\psi (\Gamma
)\leq m_{0}$. Thus, $\psi _{U}^{a}(z)\leq m_{0}$ for any problem $z\in
\mathcal{P}(U,n)$. From here and from Lemma \ref{L3.4} it follows that $typ(%
\mathcal{U}_{U\psi n}^{ai})=\alpha $.

We now show that $typ(\mathcal{U}_{U\psi n}^{di})=\gamma $. Using Lemmas \ref%
{L3.2} and \ref{L3.3}, and Corollary \ref{C3.1} we obtain $Dom(\mathcal{U}%
_{U\psi n}^{di})=\omega \setminus \{0\}$. Let $m\in \omega \setminus \{0\}$.
One can show that there exist functions $f_{1},\ldots ,f_{m}\in
P_{v_{r}}^{(n_{r})}$ such that, for any $\delta _{1},\ldots ,\delta _{m}\in
E_{2}$, the system of equations%
\[
\{f_{1}(\bar{x})=\delta _{1},\ldots ,f_{m}(\bar{x})=\delta _{m}\},
\]%
where $\bar{x}=(x_{1},\ldots ,x_{n_{r}})$, is consistent on $A^{n}$.

Consider the problem $z=(Y,\nu ,f_{1}(\bar{x}),\ldots ,f_{m}(\bar{x}))$,
where $Y=\{x_{1},\ldots ,x_{n}\}$, $\nu :E_{2}^{m}\rightarrow \mathcal{S}%
(\omega )$ and $\nu (\bar{\delta}_{1})\cap \nu (\bar{\delta}_{2})=\emptyset $
for any $\bar{\delta}_{1},\bar{\delta}_{2}\in E_{2}^{m}$ such that $\bar{%
\delta}_{1}\neq \bar{\delta}_{2}$. It is clear that $\psi _{U}^{i}(z)=m$.
Let $\Gamma $ be a computation tree over $U,$ which solves the problem $z$
deterministically and for which $\psi (\Gamma )=\psi _{U}^{d}(z)$.
Evidently, the computation tree $\Gamma $ must have at least $2^{m}$ terminal
nodes. Therefore $h(\Gamma )\geq m$ and $\psi _{U}^{d}(z)\geq m$. Thus, $%
\mathcal{U}_{U\psi n}^{di}(m)\geq m$. By Lemma \ref{L3.1}, $\mathcal{U}%
_{U\psi n}^{di}(m)=m$. Taking into account that $m$ is an arbitrary number
from $\omega \setminus \{0\}$, we obtain $typ(\mathcal{U}_{U\psi
n}^{di})=\gamma $. From this equality, from equality $typ(\mathcal{U}_{U\psi
n}^{ai})=\alpha $, and from Proposition \ref{P3.1} it follows that $%
typ_{u}(U,\psi ,n)=t_{4}$.
 \end{proof}

\begin{lemma}
\label{L5.4} Let $\tau \in \Delta _{u}$, $(U_{\tau },\psi _{\tau })=\pi
_{v_{1}}^{(n_{1})}\oplus \cdots \oplus \pi _{v_{m}}^{(n_{m})}$, $r\in
\{1,\ldots ,m\}$, $v_{r}=5$, $n\in \omega \setminus \{0\}$, $n\geq n_{r}$
and if $r<m$, then $n<n_{r+1}$. Then $typ_{u}(U_{\tau },\psi _{\tau
},n)=t_{5}$.
\end{lemma}

\begin{proof}
One can show that $r=1$ or $r=2$. In the latter case, $v_{1}=2$.
Denote $(U,\psi )=(U_{\tau },\psi _{\tau })$. Let $U=(A,F,P)$. We now show that $typ(\mathcal{%
U}_{U\psi n}^{da})=\gamma $. For an arbitrary $i\in \omega $, denote $%
c_{i}=\psi (q_{i}^{(n_{r})})$. Remind that $c_{0}=1$ and $c_{i}=i
\sum_{j=0}^{i-1}c_{j}+2$ for $i\geq 1$. Consider a problem $z_{i}=(Y,\nu
,q_{i}^{(n_{r})}(x_{1},\ldots ,x_{n_{r}}))$, where $Y=\{x_{1},\ldots
,x_{n}\} $, $\nu (0)=\{0\}$ and $\nu (1)=\{1\}$. It is clear that $\psi
_{U}^{i}(z_{i})=c_{i}$. We now show that $\psi _{U}^{a}(z_{i})=\psi
_{U}^{d}(z_{i})=c_{i}$. One can prove that there exists a computation tree $%
\Gamma $ over $U,$ which solves the problem $z_{i}$ nondeterministically and
satisfies the following conditions: $\psi (\Gamma )=\psi _{U}^{a}(z_{i})$
and all expressions attached to nodes of $\Gamma$ depend on variables from $Y$ only. We now
show that the expression $q_{i}^{(n_{r})}(x_{1},\ldots ,x_{n_{r}})$ is
attached to a node of $\Gamma $. Assume the contrary. Consider two $n$%
-tuples from $A^{n}$: $\bar{\alpha}=(k_{0}^{(n_{r})},\ldots
,k_{0}^{(n_{r})}) $ and $\bar{\beta}=(k_{1}^{(n_{r})},\ldots
,k_{n_{r}-1}^{(n_{r})},i^{(n_{r})},k_{0}^{(n_{r})},\ldots ,k_{0}^{(n_{r})})$%
. Note that all functions from $P$ depending on variables from $Y$ take the
value $0$ on the tuple $\bar{\alpha}$, and all functions from $P$ depending
on variables from $Y$ with the exception of $q_{i}^{(n_{r})}(x_{1},\ldots
,x_{n_{r}})$ take the value $0$ on the tuple $\bar{\beta}$. It is clear that
$z_{i}(\bar{\alpha})=\{0\}$ and $z_{i}(\bar{\beta})=\{1\}$. Let $\xi $ be a
complete path in $\Gamma $ such that $\bar{\alpha}\in \mathcal{A}(\xi )$.
Evidently, all edges of this path with the exception of the first one are
labeled with the number $0$. Therefore $\bar{\beta}\in \mathcal{A}(\xi )$,
but this is impossible since $z_{i}(\bar{\alpha})\cap z_{i}(\bar{\beta}%
)=\emptyset $. Hence the expression $q_{i}^{(n_{r})}(x_{1},\ldots
,x_{n_{r}}) $ is attached to at least one node of $\Gamma $. Taking into
account that $\psi _{U}^{i}(z_{i})=c_{i}$ and using Lemma \ref{L3.1} we
obtain that $\psi _{U}^{a}(z_{i})=\psi _{U}^{d}(z_{i})=c_{i}$. Hence if the
value $\mathcal{U}_{U\psi n}^{da}(c_{i})$ is definite, then it satisfies the
inequality $\mathcal{U}_{U\psi n}^{da}(c_{i})\geq c_{i}$.

Let $i\in \omega \setminus \{0\}$ and $i\geq (n)_{n_{r}}$. We now prove that
the value $\mathcal{U}_{U\psi n}^{da}(c_{i}-1)$ is definite and $\mathcal{U}%
_{U\psi n}^{da}(c_{i}-1)<c_{i}-1$. One can show that the set $\{z:z\in
\mathcal{P}(U,n),\psi _{U}^{a}(z)\le c_{i}-1\}$ is not empty. Let $z=(Y,\nu
,\beta _{1},\ldots ,\beta _{t})$ be an arbitrary problem from this set. One
can prove that there exists a computation tree $\Gamma $ over $U$, which solves
the problem $z$ nondeterministically and satisfies the following conditions:
$\psi (\Gamma )\le c_{i}-1$ and all expressions attached to nodes of $\Gamma $
depend on variables from $Y$. Denote by $\Phi (\Gamma )$ the set of all
functions that are attached to nodes of $\Gamma $ and are not identically
equal to $0$. One can show that there exists a computation tree $\Gamma ^{\prime
}$, which solves the problem $z$ deterministically by sequential computing
values of all functions from $\Phi (\Gamma )$. Since all functions from $%
\Phi (\Gamma )$ are not identically equal to $0$, there are of the kind $%
q_{j}^{(n_{r})}(x_{s_{1}},\ldots ,x_{s_{n_{r}}})$, where $x_{s_{1}},\ldots
,x_{s_{n_{r}}}$ are pairwise different variables from $Y$. Since $\psi
(\Gamma )\le c_{i}-1$, $j\leq i-1$. It is clear that $\psi (\Gamma ^{\prime
})=\sum_{f\in \Phi (\Gamma )}\psi(f)\le (n)_{n_r}\sum_{j=0}^{i-1}c_{j}$. Since $i\geq (n)_{n_{r}}$, $\psi (\Gamma ^{\prime}) \leq i\sum_{j=0}^{i-1}c_{j}=c_{i}-2<c_{i}-1$. Hence $\psi
_{U}^{d}(z)<c_{i}-1$. Taking into account that $z$ is an arbitrary problem
from $\mathcal{P}(U,n)$ such that $\psi _{U}^{a}(z)\le c_{i}-1$, we obtain that
the value $\mathcal{U}_{U\psi n}^{da}(c_{i}-1)$ is definite and satisfies
the inequality $\mathcal{U}_{U\psi n}^{da}(c_{i}-1)<c_{i}-1$. Therefore $Dom(%
\mathcal{U}_{U\psi n}^{da})$ is an infinite set.

By Lemma \ref{L3.2}, $Dom(%
\mathcal{U}_{U\psi n}^{da})=\omega $. Since, for any $i\in \omega $, the
inequality $\mathcal{U}_{U\psi n}^{da}(c_{i})\geq c_{i}$ holds, $Dom^{+}(%
\mathcal{U}_{U\psi n}^{da})$ is an infinite set.
Since, for any $i\in \omega $, $i\geq (n)_{n_{r}}$, the
inequality $\mathcal{U}_{U\psi n}^{da}(c_{i}-1)< c_{i}-1$ holds,
 the set $Dom^{-}(\mathcal{U}_{U\psi
n}^{da})$ is infinite. Therefore $typ(\mathcal{U}_{U\psi n}^{da})=\gamma $.
Using Proposition \ref{P3.1} we obtain that $typ_{u}(U,\psi ,n)=t_{5}$.
 \end{proof}

\begin{lemma}
\label{L5.5} Let $\tau \in \Delta _{u}$, $(U_{\tau },\psi _{\tau })=\pi
_{s_{1}}^{(n_{1})}\oplus \cdots \oplus \pi _{s_{m}}^{(n_{m})}$, $r\in
\{1,\ldots ,m\}$, $s_{r}=6$, $n\in \omega \setminus \{0\}$, $n\geq n_{r}$
and if $r<m$, then $n<n_{r+1}$. Then $typ_{u}(U_{\tau },\psi _{\tau
},n)=t_{6}$.
\end{lemma}

\begin{proof}
One can show that if $r>1$, then $\{s_{1},\ldots ,s_{r-1}\}\subseteq \{2,5\}$. Denote $(U,\psi )=(U_{\tau },\psi _{\tau })$. Let $U=(A,F,P)$.

Show that, for any $m\in \omega $, the value $\mathcal{U}_{U\psi n}^{da}(m)$
is definite. One can prove that the set $\{z:z\in \mathcal{P}(U,n),\psi
_{U}^{a}(z)\leq m\}$ is not empty. Let $z=(Y,\nu ,\beta _{1},\ldots ,\beta
_{t})$ be an arbitrary problem from this set. Let $\Phi _{Y}$ be the set of
all functions from $P[F]$ with variables from $Y$, which are not
identically equal to $0$ on $A^{n}$. One can show that the set $\Phi _{Y}$
contains only functions of the kind $q_{2i}^{(n_{r})}(x_{j_{1}},\ldots
,x_{j_{n_{r}}})$, $q_{2i+1}^{(n_{r})}(x_{j_{1}},\ldots ,x_{j_{n_{r}}})$, $%
p_{2i}^{(n_{r})}(x_{j_{1}},\ldots ,x_{j_{n_{r}}})$, where $x_{j_{1}},\ldots
,x_{j_{n_{r}}}$ are pairwise different variables from $Y$, and if $s_{t}=5$
for some $t\in \{1,\ldots ,r-1\}$, then the set $\Phi _{Y}$ contains also
functions of the kind $q_{i}^{(n_{t})}(x_{j_{1}},\ldots ,x_{j_{n_{t}}})$,
where $x_{j_{1}},\ldots ,x_{j_{n_{t}}}$ are pairwise different variables
from $Y$. Denote $\Phi _{Y}(m)=\{f:f\in \Phi _{Y},\psi (f)\leq m\}$. One can
show that $\Phi _{Y}(m)$ is a finite set and $\psi _{U}^{d}(z)\leq
\sum_{f\in \Phi _{Y}(m)}\psi (f)$. Taking into account that $z$ is an
arbitrary problem from $\mathcal{P}(U,n)$ such that $\psi _{U}^{a}(z)\leq m$%
, we obtain that the value $\mathcal{U}_{U\psi n}^{da}(m)$ is definite.

We now show that, for any $m\in \omega \setminus \{0\}$, the inequality $%
\mathcal{U}_{U\psi n}^{da}(m)>m$ holds. Consider the problem $z_{m}=(Y,\nu
,q_{2m}^{(n_{r})}(\bar{x}),q_{2m+1}^{(n_{r})}(\bar{x}),p_{2m}^{(n_{r})}(\bar{%
x}))$, where $Y=\{x_{1},\ldots ,x_{n}\}$, $\bar{x}=(x_{1},\ldots ,x_{n_{r}})$%
, $\nu ((1,0,1))=\{1\}$, $\nu ((0,1,1))=\{2\}$, and $\nu (\bar{\delta}%
)=\{0\} $ for any $3$-tuple $\bar{\delta}\in E_{2}^{3}\setminus
\{(1,0,1),(0,1,1)\}$. Consider the computation tree $\Gamma _{0}$ containing
exactly three complete paths $\xi _{0}$, $\xi _{1}$, and $\xi _{2}$, where $%
\xi _{i}=v_{0},d_{0i},v_{1i},d_{1i},v_{2i}$ for $i=0,1,2$. For $i=0,1,2$,
the node $v_{0}$ and the edge $d_{0i}$ are not labeled, and the node $v_{2i}$
is labeled with the number $i$. The node $v_{10}$ is labeled with the
expression $p_{2m}^{(n_{r})}(\bar{x})$ and the edge $d_{10}$ is labeled with
the number $0$. The node $v_{11}$ is labeled with the expression $%
q_{2m}^{(n_{r})}(\bar{x})$ and the edge $d_{11}$ is labeled with the number $%
1$. The node $v_{12}$ is labeled with the expression $q_{2m+1}^{(n_{r})}(%
\bar{x})$ and the edge $d_{12}$ is labeled with the number $1$. One can show
that $\Gamma _{0}$ solves the problem $z_{m}$ nondeterministically and $\psi
(\Gamma _{0})=m$. Therefore $\psi _{U}^{a}(z_{m})\leq m$.

Show that $\psi _{U}^{d}(z_{m})\geq 2m$. One can prove that there exists a
computation tree $\Gamma $ over $U$, which solves the problem $z_{m}$
deterministically and satisfies the following conditions: $\psi (\Gamma
)=\psi _{U}^{d}(z_{m})$ and all functions attached to nodes of $\Gamma $
depend on variables from the set $Y$ only. Define three $n$-tuples $\bar{%
\alpha}_{0}$, $\bar{\alpha}_{1}$, and $\bar{\alpha}_{2}$ from $A^{n}$: $\bar{\alpha}_{0}=(k_{0}^{(n_{r})},\ldots
,k_{0}^{(n_{r})})$%
, $\bar{\alpha}_{1}=(k_{1}^{(n_{r})},\ldots
,k_{n_{r}-1}^{(n_{r})},(2m)^{(n_{r})},k_{0}^{(n_{r})},\ldots
,k_{0}^{(n_{r})})$, and $\bar{\alpha}_{2}=(k_{1}^{(n_{r})},\ldots
,k_{n_{r}-1}^{(n_{r})},(2m+1)^{(n_{r})},k_{0}^{(n_{r})},\ldots
,$ $k_{0}^{(n_{r})})$. It is clear that $z_{m}(\bar{\alpha}_{0})=\{0\}$, $z_{m}(%
\bar{\alpha}_{1})=\{1\}$, and $z_{m}(\bar{\alpha}_{2})=\{2\}$. For $i=0,1,2$%
, denote by $\varphi _{i}$ the set of all functions from $\Phi _{Y}$, which
take value $1$ on the tuple $\bar{\alpha}_{i}$. One can show that $\varphi
_{0}=\emptyset $, $\varphi _{1}=\{q_{2m}^{(n_{r})}(\bar{x}),p_{2m}^{(n_{r})}(%
\bar{x})\}$, and $\varphi _{2}=\{q_{2m+1}^{(n_{r})}(\bar{x}%
),p_{2m}^{(n_{r})}(\bar{x})\}$.

Consider the complete path $\xi _{0}$ in $%
\Gamma $ such that $\bar{\alpha}_{0}\in \mathcal{A}(\xi _{0})$. It is clear
that $\bar{\alpha}_{1}\notin \mathcal{A}(\xi _{0})$ and $\bar{\alpha}%
_{2}\notin \mathcal{A}(\xi _{0})$. Therefore the expression $%
p_{2m}^{(n_{r})}(\bar{x})$ or both expressions $q_{2m}^{(n_{r})}(\bar{x})$
and $q_{2m+1}^{(n_{r})}(\bar{x})$ are among expressions attached to the nodes of
$\xi _{0}$. If nodes of $\xi _{0}$ are labeled with at least two expressions
from the set $B=\{q_{2m}^{(n_{r})}(\bar{x}),q_{2m+1}^{(n_{r})}(\bar{x}%
),p_{2m}^{(n_{r})}(\bar{x})\}$, then $\psi (\Gamma )\geq 2m$. Let only one
expression from the set $B$ be attached to the nodes of $\xi _{0}$. Then this is
$p_{2m}^{(n_{r})}(\bar{x})$. Let the considered expression be attached to
the node $v$ of the path $\xi _{0}$. Consider the complete path $\xi _{1}$ in $\Gamma$
such that $\bar{\alpha}_{1}\in \mathcal{A}(\xi _{1})$. Since $\Gamma $ is a
deterministic computation tree, the path $\xi _{1}$ contains the node $v$.
Suppose that no one node of the path $\xi _{1}$ is labeled with an
expression from the set $\{q_{2m}^{(n_{r})}(\bar{x}),q_{2m+1}^{(n_{r})}(\bar{%
x})\}$. Then $\bar{\alpha}_{2}\in \mathcal{A}(\xi _{1})$ but this is
impossible. Hence nodes of $\xi _{1}$ are labeled with at least two
expressions from the set $B$. Therefore $\psi (\Gamma )\geq 2m$ and $\psi
_{U}^{d}(z_{m})\geq 2m$. Taking into account that $\psi _{U}^{a}(z_{m})\leq
m $, we obtain $\mathcal{U}_{U\psi n}^{da}(m)\geq 2m$ for any $m\in \omega
\setminus \{0\}$. Thus, $typ(\mathcal{U}_{U\psi n}^{da})=\delta $. Using
Proposition \ref{P3.1} we obtain $typ_{u}(U,\psi ,n)=t_{6}$.
 \end{proof}

\begin{lemma}
\label{L5.6} Let $\tau \in \Delta _{u}$, $(U_{\tau },\psi _{\tau })=\pi
_{s_{1}}^{(n_{1})}\oplus \cdots \oplus \pi _{s_{m}}^{(n_{m})}$,  $s_{m}=7$, $n\in \omega \setminus \{0\}$, and $n\geq n_{m}$%
. Then $typ_{u}(U_{\tau },\psi _{\tau },n)=t_{7}$.
\end{lemma}

\begin{proof}
Denote $(U,\psi )=(U_{\tau },\psi _{\tau })$. Let $U=(A,F,P)$. Show that $%
typ(\mathcal{U}_{U\psi n}^{da})=\varepsilon $. Let $t\in \omega \setminus
\{0\}$. Consider the problem $z_{t}=(Y,\nu ,l_{1}^{(n_{m})}(\bar{x}),\ldots
,l_{t}^{(n_{m})}(\bar{x}))$ from $\mathcal{P}(U,n)$, where $Y=\{x_{1},\ldots
,x_{n}\}$, $\bar{x}=(x_{1},\ldots ,x_{n_{m}})$, $\nu :E_{2}^{t}\rightarrow
\mathcal{S}(\omega )$, and $\nu (\bar{\delta}_{1})\cap \nu (\bar{\delta}%
_{2})=\emptyset $ for any $\bar{\delta}_{1},\bar{\delta}_{2}\in E_{2}^{t}$
such that $\bar{\delta}_{1}\neq \bar{\delta}_{2}$. It is not difficult to
show that, for any consistent on $A^{n}$ system of equations%
\[
\{l_{1}^{(n_{m})}(\bar{x})=\delta _{1},\ldots ,l_{t}^{(n_{m})}(\bar{x}%
)=\delta _{t}\},
\]%
where $\delta _{1},\ldots ,\delta _{t}\in E_{2}$, there exists a subsystem,
which has the same set of solutions and which contains at most two
equations. Using this fact it is not difficult to show that $\psi
_{U}^{a}(z_{t})\leq 2$.

Let us prove that there is no $c\in \omega $ such
that $\psi _{U}^{d}(z_{t})\leq c$ for any $t\in \omega \setminus \{0\}$.
Assume the contrary. Then there exists a number $w\in \omega \setminus \{0\}$
satisfying the following condition: for any $t\in \omega \setminus \{0\}$,
there exists a computation tree $\Gamma $ over $U,$ which solves the problem $%
z_{t}$ deterministically and has at most $w$ terminal nodes. Let us consider
the problem $z_{w}$. Let $\Gamma $ be an arbitrary computation tree over $U$,
which solves the problem $z_{w}$ deterministically. It is not difficult to
show that $\Gamma $ must have at least $w+1$ terminal nodes. We obtain a
contradiction. Thus, $\psi _{U}^{a}(z_{t})\leq 2$ for any $t\in \omega
\setminus \{0\}$ and there is no $c\in \omega $ such that $\psi
_{U}^{d}(z_{t})\leq c$ for any $t\in \omega \setminus \{0\}$. Therefore $%
\mathcal{U}_{U\psi n}^{da}(2)=\infty $. Using Lemma \ref{L3.2} we obtain $%
typ(\mathcal{U}_{U\psi n}^{da})=\varepsilon $.

We now prove that the function $\psi _{U}^{a}$ is unbounded from above on
the set $\mathcal{P}(U,n)$. Let $i\in \omega \setminus \{0\}$. Consider the
problem $\eta _{i}=(Y,\nu ,q_{-i}^{(n_{m})}(\bar{x}))$ over $U$, where $%
Y=\{x_{1},\ldots ,x_{n}\}$, $\bar{x}=(x_{1},\ldots ,x_{n_{m}})$, $\nu
((0))=\{0\}$, and $\nu ((1))=\{1\}$. Show that $\psi _{U}^{a}(\eta _{i})\geq
i.$ One can prove that there exists a computation tree $\Gamma $ over $U$, which
solves the problem $\eta _{i}$ nondeterministically and satisfies the
following conditions: $\psi (\Gamma )=\psi _{U}^{a}(\eta _{i})$, and all
expressions attached to nodes of $\Gamma $ depend on variables from $Y$
only. Show that the expression $q_{-i}^{(n_{m})}(\bar{x})$ is attached to a
node of $\Gamma $. Assume the contrary. Consider two $n$-tuples from $A^{n}$%
: $\bar{\alpha}=(k_{0}^{(n_{m})},\ldots ,k_{0}^{(n_{m})})$ and $\bar{\beta}%
=(k_{1}^{(n_{m})},\ldots
,k_{n_{m}-1}^{(n_{m})},(-i)^{(n_{m})},k_{0}^{(n_{m})},\ldots
,k_{0}^{(n_{m})})$.
It is clear that $\eta _{i}(\bar{\alpha})=\{0\}$ and $\eta _{i}(\bar{\beta}%
)=\{1\}$.
One can show that all functions from $P$ depending on
variables from $Y$ take the value $0$ on the tuple $\bar{\alpha}$, and
all functions from $P$, depending on variables from $Y$, with the exception
of $q_{-i}^{(n_{m})}(\bar{x})$, take the value $0$ on the tuple $\bar{\beta}$%
.  Let $\xi $ be a complete path in $\Gamma $ such that $\bar{\alpha}%
\in \mathcal{A}(\xi )$. Evidently, all edges of this path with the exception
of the first one are labeled with $0$. Therefore $\bar{\beta}\in \mathcal{A}%
(\xi )$ but this is impossible since $\eta _{i}(\bar{\alpha})\cap \eta _{i}(%
\bar{\beta})=\emptyset $. Hence the expression $q_{-i}^{(n_{m})}(\bar{x})$
is attached to a node of $\Gamma $. Therefore $\psi (\Gamma )\geq i$ and $%
\psi _{U}^{a}(\eta _{i})\geq i$. Thus, the function $\psi _{U}^{a}$ is
unbounded from above on the set $\mathcal{P}(U,n)$. Using Lemma \ref{L3.4}
we obtain $typ(\mathcal{U}_{U\psi n}^{aa})=\gamma $. From this equality, the
equality $typ(\mathcal{U}_{U\psi n}^{da})=\varepsilon $, and from Proposition %
\ref{P3.1} it follows that $typ_{u}(U,\psi ,n)=t_{7}$.
 \end{proof}

\begin{lemma}
\label{L5.7} Let $(U,\psi )$ be a sm-pair such that $\psi \equiv 0$. Then $%
dtyp_{u}(U,\psi )=t_{1}^{\infty }$.
\end{lemma}

\begin{proof}
Let $n\in \omega \setminus \{0\}$. It is clear that the function $\psi
_{U}^{i}$ is bounded from above on the set $\mathcal{P}(U,n)$. Using Lemma %
\ref{L3.4} we obtain $typ(\mathcal{U}_{U\psi n}^{ii})=\alpha $. From this
equality and from Proposition \ref{P3.1} it follows that $typ_{u}(U,\psi
,n)=t_{1}$.
 \end{proof}

\begin{proof}
[Proof of Proposition \ref{P5.1}] Let $\tau \in \{t_{1}^{\infty }\}\cup \Delta _{u}$%
. Assume that $\tau =t_{1}^{\infty }$. Consider an arbitrary sm-pair $%
(U,\psi )$ such that $\psi \equiv 0$. By Lemma \ref{L5.7}, $dtyp_{u}(U,\psi
)=t_{1}^{\infty }$. Assume now that $\tau \in \Delta _{u}$. Consider the
sm-pair $(U_{\tau },\psi _{\tau})$. Using Lemmas \ref{L5.1}--\ref{L5.6} we
obtain $dtyp_{u}(U_{\tau },\psi _{\tau })=\tau $.
 \end{proof}

\begin{proof}
[Proof of Proposition \ref{P5.2}] Let $\tau \in \Delta _{u}$. Consider the sm-pair $(U_{\tau },\psi _{\tau})$. 	
By construction, $\psi _{\tau}$ is a weighted depth.  Therefore $(U_{\tau },\psi _{\tau })$ is a limited sm-pair.
From Lemmas \ref{L5.1}--\ref{L5.6} it follows that $%
dtyp_{u}(U_{\tau },\psi _{\tau })=\tau $.
 \end{proof}

\section{Proofs of Theorems \protect\ref{T2.1} and \protect\ref{T2.2}}
\label{S7}

Let us define a function $\rho :\{\alpha ,\beta ,\gamma ,\delta ,\varepsilon
\}\rightarrow \{\alpha ,\beta ,\gamma ,\delta ,\varepsilon \}$ as follows: $%
\rho (\alpha )=\varepsilon $, $\rho (\beta )=\delta $, $\rho (\gamma
)=\gamma $, $\rho (\delta )=\beta $, and $\rho (\varepsilon )=\alpha $. The
following statement (Proposition 5 from \cite{Moshkov96}) allows us to analyze the relationships
between upper types and types of sm-pairs.

\begin{proposition}
 \label{P6.1} Let $B$ be a nonempty set, $f:B\rightarrow
\omega $, $g:B\rightarrow \omega $, $\mathcal{U}^{fg}(n)=\max \{f(b):b\in
B,g(b)\leq n\}$, and $\mathcal{L}^{gf}(n)=\min \{g(b):b\in B,f(b)\geq n\}$
for any $n\in \omega $. Then $typ(\mathcal{L}^{gf})=\rho (typ(\mathcal{U}%
^{fg}))$.
\end{proposition}

Using Proposition \ref{P6.1} we obtain the following statement.

\begin{proposition}
\label{P6.2} Let $(U,\psi )$ be a sm-pair, $n\in \omega \setminus \{0\}$,
and $typ_{u}(U,\psi ,n)=t_{i}$ for some $i\in \{1,\ldots ,7\}$. Then $%
typ(U,\psi ,n)=T_{i}$.
\end{proposition}

\begin{proof}
[Proof of Theorem \ref{T2.1}] The statement of the theorem follows from
Propositions \ref{P4.1}, \ref{P5.1}, and \ref{P6.2}.
 \end{proof}

\begin{proof}
[Proof of Theorem \ref{T2.2}] The statement of the theorem follows from
Propositions \ref{P4.2}, \ref{P5.2}, and \ref{P6.2}.
 \end{proof}

\section{Explanations of Results}
\label{S8}

Rather unusual formulations of the results (in particular, the use of the types of functions) require additional explanations.

Let $(U,\psi )$ be a sm-pair, and $n\in \omega \setminus \{0\}$. We'll start
by considering the upper $n$-type $typ_{u}(U,\psi ,n)$ of the sm-pair $%
(U,\psi )$. After that, discuss the lower $n$-type $typ_{l}(U,\psi ,n)$ of the
sm-pair $(U,\psi )$ and the $n$-type $typ(U,\psi ,n)$ of the sm-pair $(U,\psi )$%
. We end by looking at the dynamic type $dtyp(U,\psi )$ of the sm-pair $(U,\psi )
$.

\subsection{Upper $n$-Types of SM-Pairs}

For any $b,c\in \{i,d,a\}$, the matrix $typ_{u}(U,\psi ,n)$ (the upper $n$%
-type of the sm-pair $(U,\psi )$) contains the value $typ(\mathcal{U}_{U\psi
n}^{bc})\in \{\alpha ,\beta ,\gamma ,\delta ,\varepsilon \}$ in the
intersection of the row $b$ and the column $c$, where $$\mathcal{U}_{U\psi
n}^{bc}(m)=\max \{\psi _{U}^{b}(z):z\in \mathcal{P}(U,n),\psi
_{U}^{c}(z)\leq m\}$$ for any $m\in \omega $. If the value $\mathcal{U}%
_{U\psi n}^{bc}(m)$ is defined for some $m\in \omega $, then it is the
unimprovable upper bound on the value $\psi _{U}^{b}(z)$ for problems $z\in
\mathcal{P}(U,n)$ such that the inequality $\psi _{U}^{c}(z)\leq m$ holds.
From Propositions \ref{P4.1},  \ref{P4.2}, \ref{P5.1}, and \ref{P5.2} it follows that $\{t_{1},\ldots
,t_{7}\}$ is the set of all possible upper $n$-types of sm-pairs and $%
\{t_{2},\ldots ,t_{7}\}$ is the set of all possible upper $n$-types of limited
sm-pairs.

Let us remind that $Dom(\mathcal{U}_{U\psi n}^{bc})$ is the domain of $%
\mathcal{U}_{U\psi n}^{bc}$, $Dom^{+}(\mathcal{U}_{U\psi n}^{bc})=\{m:m\in
Dom(\mathcal{U}_{U\psi n}^{bc}),\mathcal{U}_{U\psi n}^{bc}(m)\geq m\}$, and $%
Dom^{-}(\mathcal{U}_{U\psi n}^{bc})=\{m:m\in Dom(\mathcal{U}_{U\psi n}^{bc}),%
\mathcal{U}_{U\psi n}^{bc}(m)\leq m\}$.

By Lemma \ref{L3.2}, the set $Dom(\mathcal{U}_{U\psi n}^{bc})$ is finite if
and only if there exists $m\in \omega $ such that $\mathcal{U}_{U\psi
n}^{bc}(m)=\infty $, i.e., the set $\{\psi _{U}^{b}(z):z\in \mathcal{P}%
(U,n),\psi _{U}^{c}(z)\leq m\}$ is infinite. In this case, $Dom(\mathcal{U}%
_{U\psi n}^{bc})=\emptyset $ if $m_{0}=m_{1}$ and $Dom(\mathcal{U}_{U\psi
n}^{bc})=\{m:m\in \omega ,m_{0}\leq m<m_{1}\}$ if $m_{0}<m_{1}$, where $%
m_{0}=\min \{\psi _{U}^{c}(z):z\in \mathcal{P}(U,n)\}$ and $m_{1}=\min
\{m:m\in \omega ,\mathcal{U}_{U\psi n}^{bc}(m)=\infty \}$. If the set $Dom(%
\mathcal{U}_{U\psi n}^{bc})$ is infinite, then $Dom(\mathcal{U}_{U\psi
n}^{bc})=\{m:m\in \omega ,m\geq m_{0}\}$.

Let us remind that, by Lemma \ref{L3.1}, $\psi _{U}^{a}(z)\leq \psi
_{U}^{d}(z)\leq \psi _{U}^{i}(z)$ for any problem $z\in \mathcal{P}(U,n)$.

The equality $typ(\mathcal{U}_{U\psi n}^{bc})=\alpha $ means that the set $Dom(%
\mathcal{U}_{U\psi n}^{bc})$ is infinite and the function $\mathcal{U}%
_{U\psi n}^{bc}$ is bounded from above. This equality can hold for any pair $%
bc,$ $b,c\in \{i,d,a\}$. By Lemma \ref{L3.4}, $typ(\mathcal{U}_{U\psi
n}^{bc})=\alpha $  if and only if the function $\psi _{U}^{b}$ is bounded
from above on the set $\mathcal{P}(U,n)$. Later we will often omit words
\textquotedblleft from above on the set $\mathcal{P}(U,n)$%
\textquotedblright\ and write that the function $\psi _{U}^{b}$ is bounded
or that the function $\psi _{U}^{b}$ is unbounded. The function $\psi
_{U}^{i}$ can be bounded only for sm-pairs that are not limited. The
function $\psi _{U}^{d}$ can be bounded for sm-pairs that are limited, but
this case is in some sense degenerate -- see Lemma \ref{L0.1}.

The equality $typ(\mathcal{U}_{U\psi n}^{bc})=\beta $ means that $Dom(%
\mathcal{U}_{U\psi n}^{bc})$ is an infinite set, $Dom^{+}(\mathcal{U}_{U\psi
n}^{bc})$ is a finite set, and $\mathcal{U}_{U\psi n}^{bc}$ is an unbounded
from above function. This equality can hold only if $bc=di$. One can show
that $typ(\mathcal{U}_{U\psi n}^{di})=\beta $ if and only if
the function $\psi _{U}^{d}$ is unbounded
and there exists $%
p\in \omega $ such that $\psi _{U}^{d}(z)<\psi _{U}^{i}(z)$ for any problem $%
z\in \mathcal{P}(U,n)$ such that $\psi _{U}^{i}(z)\geq p$.

The equality $typ(\mathcal{U}_{U\psi n}^{bc})=\gamma $ means that each of
the sets $Dom^{+}(\mathcal{U}_{U\psi n}^{bc})$ and $Dom^{-}(\mathcal{U}%
_{U\psi n}^{bc})$ is an infinite set. This equality can hold for any pair $%
bc,$ $b,c\in \{i,d,a\}$, with the exception of $id$ and $ia$.
One can show that, for $bc \in \{di, ai, ad\}$,
$typ(\mathcal{U}_{U\psi n}^{bc})=\gamma $ if and only if, for any $q\in
\omega $, there exists a problem $z\in \mathcal{P}(U,n)$ such that $\psi
_{U}^{b}(z)=\psi _{U}^{c}(z)\geq q$. If $bc \in \{ii, dd, aa\}$, then $typ(\mathcal{U}_{U\psi n}^{bc})=\gamma $ if and only if the function
$\psi _{U}^{b}$ is unbounded -- see Lemma \ref{L3.4}. One can show that
$typ(\mathcal{U}_{U\psi n}^{da})=\gamma $ if and only if the set $Dom(\mathcal{U}_{U\psi n}^{da})$ is infinite, the function $\psi _{U}^{d}$ is unbounded and,
for any $q\in
\omega $, there exists a number $m\in\omega $, $m\ge q$, such that, for any problem $z\in \mathcal{P}(U,n)$ with $\psi
_{U}^{a}(z) \le m$,  the inequality $\psi _{U}^{d}(z)\le m$ holds.

The equality $typ(\mathcal{U}_{U\psi n}^{bc})=\delta $ means that $Dom(%
\mathcal{U}_{U\psi n}^{bc})$ is an infinite set and $Dom^{-}(\mathcal{U}%
_{U\psi n}^{bc})$ is a finite set. This equality can hold only if $bc=da$.
One can show that $typ(\mathcal{U}_{U\psi n}^{da})=\delta $ if and only if $Dom(\mathcal{U}_{U\psi n}^{da})$ is an infinite set  and
there exists $p\in \omega $ such that, for any $m\in \omega $, $m \ge p$, there exists a problem  $z\in \mathcal{P}(U,n)$ such that $\psi _{U}^{a}(z)\le m$ and $\psi _{U}^{d}(z) >m$.

The equality $typ(\mathcal{U}_{U\psi n}^{bc})=\varepsilon $ means that $Dom(%
\mathcal{U}_{U\psi n}^{bc})$ is a finite set. This equality can hold only if
$bc\in \{id,ia,da\}$. By Lemma \ref{L3.2}, $typ(\mathcal{U}_{U\psi
n}^{bc})=\varepsilon $ if and only if there exists $m\in \omega $ such that $%
\mathcal{U}_{U\psi n}^{bc}(m)=\infty $, i.e., the set $\{\psi
_{U}^{b}(z):z\in \mathcal{P}(U,n),\psi _{U}^{c}(z)\leq m\}$ is infinite. The
case $typ(\mathcal{U}_{U\psi n}^{bc})=\varepsilon $ is not very informative:
for large enough $m$, there are problems $z\in \mathcal{P}(U,n)$ for which $%
\psi _{U}^{c}(z)\leq m$ but we cannot derive any upper bound on the value $%
\psi _{U}^{b}(z)$ for these problems.

Let $typ_{u}(U,\psi ,n)=t\in \{t_{1},\ldots ,t_{7}\}$. We now consider possible behavior of the rows in the table $t$, which is closely related to the boundedness of the functions $\psi _{U}^{i}$, $\psi _{U}^{d}$, and $\psi _{U}^{a}$.

The row $i$ in the
matrix $t$ is equal to $(\alpha ,\alpha ,\alpha )$ if the function $\psi
_{U}^{i}$ is bounded (if $t=t_{1}$) and is equal to $(\gamma ,\varepsilon
,\varepsilon )$ if the function $\psi _{U}^{i}$ is unbounded (if $t\in
\{t_{2},\ldots ,t_{7}\}$).

The row $a$ in the matrix $t$ is equal to $(\alpha ,\alpha ,\alpha )$ if the
function $\psi _{U}^{a}$ is bounded (if $t\in \{t_{1},t_{2},t\,_{3},t_{4}\}$%
) and is equal to $(\gamma ,\gamma ,\gamma )$ if the function $\psi _{U}^{a}$
is unbounded (if $t\in \{t_{5},t_{6},t_{7}\}$).

The behavior of the row $d$ in the matrix $t$ is more complicated. If the
function $\psi _{U}^{d}$ is bounded (if $t\in \{t_{1},t_{2}\}$), then this
row is equal to $(\alpha ,\alpha ,\alpha )$. If the function $\psi _{U}^{d}$
is unbounded and the function $\psi _{U}^{a}$ is bounded (if $t\in
\{t_{3},t_{4}\}$), then the row $d$ is equal to $(x,\gamma ,\varepsilon )$,
where $x\in \{\beta ,\gamma \}$. If each of the functions $\psi _{U}^{i}$, $%
\psi _{U}^{d}$, and $\psi _{U}^{a}$ is unbounded (if $t\in
\{t_{5},t_{6},t_{7}\}$), then the row $d$ is equal to $(\gamma ,\gamma ,x)$,
where $x\in \{\gamma ,\delta ,\varepsilon \}$.

\subsection{Lower $n$-Types of SM-Pairs}

We denote by $typ_{l}(U,\psi ,n)$ a table with three rows and three columns
in which rows from top to bottom and columns from the left to the right are
labeled with indices $i,d,a$ and the value $typ(\mathcal{L}_{U\psi n}^{bc})$
is in the intersection of the row with index $b\in \{i,d,a\}$ and the column
with index $c\in \{i,d,a\}$. The table $typ_{l}(U,\psi ,n)$ will be called the
lower $n$-type of the sm-pair $(U,\psi )$.

For any $b,c\in \{i,d,a\}$ and $m\in \omega $, $$\mathcal{L}_{U\psi
n}^{bc}(m)=\min \{\psi _{U}^{b}(z):z\in \mathcal{P}(U,n),\psi
_{U}^{c}(z)\geq m\}.$$ If the value $\mathcal{L}_{U\psi n}^{bc}(m)$ is
defined for some $m\in \omega $, then it is the unimprovable lower bound on
the value $\psi _{U}^{b}(z)$ for problems $z\in \mathcal{P}(U,n)$ such that
the inequality $\psi _{U}^{c}(z)\geq m$ holds.

One can show that the set $Dom(\mathcal{L}_{U\psi
n}^{bc})$ is finite if and only if the function $\psi _{U}^{c}$ is bounded. In this case, $Dom(\mathcal{L}%
_{U\psi n}^{bc})=\{0,\ldots ,M\}$, where $M=\max \{\psi _{U}^{c}(z):z\in
\mathcal{P}(U,n)\}$. If the function $\psi _{U}^{c}$ is unbounded, then $Dom(\mathcal{L}_{U\psi
n}^{bc})=\omega $.

Define seven tables:%
$$
l_{1}=%
\begin{tabular}{|c|ccc|}
\hline
& $i$ & $d$ & $a$ \\ \hline
$i$ & $\varepsilon $ & $\varepsilon $ & $\varepsilon $ \\
$d$ & $\varepsilon $ & $\varepsilon $ & $\varepsilon $ \\
$a$ & $\varepsilon $ & $\varepsilon $ & $\varepsilon $ \\ \hline
\end{tabular}%
\;\;l_{2}=%
\begin{tabular}{|c|ccc|}
\hline
& $i$ & $d$ & $a$ \\ \hline
$i$ & $\gamma $ & $\varepsilon $ & $\varepsilon $ \\
$d$ & $\alpha $ & $\varepsilon $ & $\varepsilon $ \\
$a$ & $\alpha $ & $\varepsilon $ & $\varepsilon $ \\ \hline
\end{tabular}%
\;\;l_{3}=%
\begin{tabular}{|c|ccc|}
\hline
& $i$ & $d$ & $a$ \\ \hline
$i$ & $\gamma $ & $\delta $ & $\varepsilon $ \\
$d$ & $\alpha $ & $\gamma $ & $\varepsilon $ \\
$a$ & $\alpha $ & $\alpha $ & $\varepsilon $ \\ \hline
\end{tabular}%
$$
$$ l_{4}=%
\begin{tabular}{|c|ccc|}
\hline
& $i$ & $d$ & $a$ \\ \hline
$i$ & $\gamma $ & $\gamma $ & $\varepsilon $ \\
$d$ & $\alpha $ & $\gamma $ & $\varepsilon $ \\
$a$ & $\alpha $ & $\alpha $ & $\varepsilon $ \\ \hline
\end{tabular}%
\;\;
l_{5}=%
\begin{tabular}{|c|ccc|}
\hline
& $i$ & $d$ & $a$ \\ \hline
$i$ & $\gamma $ & $\gamma $ & $\gamma $ \\
$d$ & $\alpha $ & $\gamma $ & $\gamma $ \\
$a$ & $\alpha $ & $\gamma $ & $\gamma $ \\ \hline
\end{tabular}%
\;\;l_{6}=%
\begin{tabular}{|c|ccc|}
\hline
& $i$ & $d$ & $a$ \\ \hline
$i$ & $\gamma $ & $\gamma $ & $\gamma $ \\
$d$ & $\alpha $ & $\gamma $ & $\gamma $ \\
$a$ & $\alpha $ & $\beta $ & $\gamma $ \\ \hline
\end{tabular}%
$$
$$
 l_{7}=%
\begin{tabular}{|c|ccc|}
\hline
& $i$ & $d$ & $a$ \\ \hline
$i$ & $\gamma $ & $\gamma $ & $\gamma $ \\
$d$ & $\alpha $ & $\gamma $ & $\gamma $ \\
$a$ & $\alpha $ & $\alpha $ & $\gamma $ \\ \hline
\end{tabular}%
$$

Using Proposition \ref{P6.1} we obtain that $\{l_{1},\ldots ,l_{7}\}$ is the
set of all possible lower $n$-types of sm-pairs and $\{l_{2},\ldots ,l_{7}\}$
is the set of all possible lower $n$-types of limited sm-pairs. Moreover, $%
typ_{l}(U,\psi ,n)=l_{i}$, $i\in \{1,\ldots ,7\}$, if and only if $%
typ_{u}(U,\psi ,n)=t_{i}$.

From Proposition \ref{P6.1} it follows that, for any pair $bc$, $b,c \in \{i,d,a\}$, $$typ(\mathcal{L}_{U\psi
n}^{bc})=\rho (typ(\mathcal{U}_{U\psi n}^{cb})),$$ where $\rho (\alpha
)=\varepsilon $, $\rho (\beta )=\delta $, $\rho (\gamma )=\gamma $, $\rho
(\delta )=\beta $, and $\rho (\varepsilon )=\alpha $. Using this equality and criteria of the behavior of the value $typ(\mathcal{U}_{U\psi n}^{cb})$ described in the previous section we can obtain the criteria of the behavior of the value $typ(\mathcal{L}_{U\psi n}^{bc})$. As examples, we consider the criteria for equalities $typ(\mathcal{L}_{U\psi n}^{bc})= \alpha$ and  $typ(\mathcal{L}_{U\psi n}^{bc})= \varepsilon$.

The equality $typ(\mathcal{L}_{U\psi n}^{bc})=\alpha $ means that the set $Dom(%
\mathcal{L}_{U\psi n}^{bc})$ is infinite and the function $\mathcal{L}%
_{U\psi n}^{bc}$ is bounded from above. This equality can hold  if $bc\in
\{di,ai,ad\}$. One can show that $typ(\mathcal{L}_{U\psi
n}^{bc})=\alpha $  if and only if there exists $p\in \omega $ such that $%
\mathcal{U}_{U\psi n}^{cb}(p)=\infty $, i.e., the set $\{\psi
_{U}^{c}(z):z\in \mathcal{P}(U,n),\psi _{U}^{b}(z)\leq p\}$ is infinite. In
this case, $\mathcal{L}_{U\psi n}^{bc}(m)\leq p$ for any $m\in \omega $.

The equality $typ(\mathcal{L}_{U\psi n}^{bc})=\varepsilon $ means that $Dom(%
\mathcal{L}_{U\psi n}^{bc})$ is a finite set. This equality can hold for any
pair $bc,$ $b,c\in \{i,d,a\}$. One can show that $typ(\mathcal{L}_{U\psi
n}^{bc})=\varepsilon $ if and only if the function $\psi _{U}^{c}$ is
bounded. The case $typ(\mathcal{L}%
_{U\psi n}^{bc})=\varepsilon $ is not as bad as the case $typ(\mathcal{U}%
_{U\psi n}^{bc})=\varepsilon $. In the former case, the value $\mathcal{L}%
_{U\psi n}^{bc}(m)$ is undefined only if there are no problems $z\in
\mathcal{P}(U,n)$ such that $\psi _{U}^{c}(z)\geq m$. In the latter case, if
the value $\mathcal{U}_{U\psi n}^{bc}(m)$ is undefined, then there exist
problems $z\in \mathcal{P}(U,n)$ with $\psi _{U}^{c}(z)\leq m$.

Let $typ_{l}(U,\psi ,n)=l\in \{l_{1},\ldots ,l_{7}\}$. We now consider possible behavior of the columns in the table $l$.

The column $i$ in the
matrix $l$ is equal to $(\varepsilon ,\varepsilon ,\varepsilon )^{T}$ if the
function $\psi _{U}^{i}$ is bounded (if $l=l_{1}$) and is equal to $(\gamma
,\alpha ,\alpha )^{T}$ if the function $\psi _{U}^{i}$ is unbounded (if $%
l\in \{l_{2},\ldots ,l_{7}\}$).

The column $a$ in the matrix $l$ is equal to $(\varepsilon ,\varepsilon
,\varepsilon )^{T}$ if the function $\psi _{U}^{a}$ is bounded (if $l\in
\{l_{1},l_{2},l\,_{3},l_{4}\}$) and is equal to $(\gamma ,\gamma ,\gamma
)^{T}$ if the function $\psi _{U}^{a}$ is unbounded (if $l\in
\{l_{5},l_{6},l_{7}\}$).

The behavior of the column $d$ in the matrix $l$ is more complicated. If the
function $\psi _{U}^{d}$ is bounded (if $l\in \{l_{1},l_{2}\}$), then this
column is equal to $(\varepsilon ,\varepsilon ,\varepsilon )^{T}$. If the
function $\psi _{U}^{d}$ is unbounded and the function $\psi _{U}^{a}$ is
bounded (if $l\in \{l_{3},l_{4}\}$), then the column $d$ is equal to $%
(x,\gamma ,\alpha )^{T}$, where $x\in \{\delta ,\gamma \}$. If each of the
functions $\psi _{U}^{i}$, $\psi _{U}^{d}$, and $\psi _{U}^{a}$ is unbounded
(if $l\in \{l_{5},l_{6},l_{7}\}$), then the column $d$ is equal to $(\gamma
,\gamma ,x)^{T}$, where $x\in \{\gamma ,\beta ,\alpha \}$.

\subsection{$n$-Types of SM-Pairs}

From Theorems \ref{T2.1} and \ref{T2.2} it follows that $\{T_{1},\ldots
,T_{7}\}$ is the set of all possible $n$-types of sm-pairs and $%
\{T_{2},\ldots ,T_{7}\}$ is the set of all possible $n$-types of limited
sm-pairs.

For $i\in \{1,\ldots ,n\}$ and $b,c\in $ $\{i,d,a \}$, we denote by $t_{i}^{bc}$ the value in the intersection of the
row $b$ and the column $c$ in the matrix $t_{i}$ and by $l_{i}^{bc}$ the value
in the intersection of the row $b$ and the column $c$ in the matrix $l_{i}$.
Then the matrix $T_{i}$ has the pair $l_{i}^{bc}t_{i}^{bc}$ in the
intersection of the row $b$ and the column $c$.

In the table
\[
\begin{tabular}{|c|ccc|}
\hline
& $i$ & $d$ & $a$ \\ \hline
$i$ & $\gamma \gamma ,\varepsilon \alpha $ & $\gamma \varepsilon ,\delta
\varepsilon ,\varepsilon \alpha ,\varepsilon \varepsilon $ & $\gamma
\varepsilon ,\varepsilon \alpha ,\varepsilon \varepsilon $ \\
$d$ & $\alpha \alpha ,\alpha \beta ,\alpha \gamma ,\varepsilon \alpha $ & $%
\gamma \gamma ,\varepsilon \alpha $ & $\gamma \gamma ,\gamma \delta ,\gamma
\varepsilon ,\varepsilon \alpha ,\varepsilon \varepsilon $ \\
$a$ & $\alpha \alpha ,\alpha \gamma ,\varepsilon \alpha $ & $\alpha \alpha
,\alpha \gamma ,\beta \gamma ,\gamma \gamma ,\varepsilon \alpha $ & $\gamma
\gamma ,\varepsilon \alpha $ \\ \hline
\end{tabular}%
\]%
in the intersection of the row with index $b\in \{i,d,a\}$ and the column
with index $c\in \{i,d,a\}$, we have all possible pairs that appear in
tables $T_{1},\ldots ,T_{7}$ in the intersection of the row and the column
with the same indices. Note that out of 25 different pairs $pq,$ $p,q\in
\{\alpha ,\beta ,\gamma ,\delta ,\varepsilon \}$, only ten pairs $\alpha
\alpha $, $\alpha \beta $, $\alpha \gamma $, $\beta \gamma $, $\gamma \gamma
$, $\gamma \delta $, $\gamma \varepsilon $, $\delta \varepsilon $, $%
\varepsilon \alpha $, $\varepsilon \varepsilon $ are present in tables $%
T_{1},\ldots ,T_{7}$.

The situation, when $typ(\mathcal{L}_{U\psi n}^{bc})typ(\mathcal{U}_{U\psi
n}^{bc})\in \{\alpha \alpha ,\beta \gamma ,\gamma \gamma ,\gamma \delta \}$
is good enough: the difference between the lower and upper bounds $\mathcal{L%
}_{U\psi n}^{bc}$ and $\mathcal{U}_{U\psi n}^{bc}$ is reasonable (they have
infinite domains and  are either both bounded from above or both unbounded
from above). For the rest of the cases, the situation is worse.
For pairs $\alpha \beta $ and $\alpha \gamma $, there is a too big gap
between lower and upper bounds $\mathcal{L}_{U\psi n}^{bc}$ and $\mathcal{U}%
_{U\psi n}^{bc}$: the lover bound is bounded from above and the upper bound
is unbounded from above. For pairs $\gamma \varepsilon $, $\delta
\varepsilon $, $\varepsilon \alpha $, and $\varepsilon \varepsilon $, at
least one of the bounds has finite domain.

\subsection{Dynamic Types for SM-Pairs}

We defined the linear order $\preceq $ on the set $\{\alpha ,\beta ,\gamma
,\delta ,\varepsilon \}$ as follows: $\alpha \preceq \beta \preceq \gamma
\preceq \delta \preceq \varepsilon $. Based on this order, we defined the
partial order  $\leqq $ on the set of tables $\{t_{1},\ldots ,t_{7}\}$. Let $i,j\in \{1,\ldots ,7\}$. Then $t_{i}\leqq t_{j}$
if and only if $t_{i}^{bc}\preceq t_{j}^{bc}$ for any $b,c\in \{i,d,a\}$.
The graph depicted in Fig. \ref{Fig1} is the Hasse diagram for the partially ordered
set $(\{t_{1},\ldots ,t_{7}\},\leqq )$. Nodes of this diagram are tables $%
t_{1},\ldots ,t_{7}$. An edge goes upward from $t_{i}$ to $t_{j}$ if $%
t_{i}\leqq t_{j}$ and there is no $t_{k}$, $t_{k}\notin \{t_{i},t_{j}\}$,
such that $t_{i}\leqq t_{k}\leqq t_{j}$. We now define a partial order $%
\leqq $ on the set of tables $\{T_{1},\ldots ,T_{7}\}$:  for any $i,j\in
\{1,\ldots ,7\}$, $T_{i}\leqq T_{j}$ if and only if  $t_{i}\leqq t_{j}$.

From Theorems \ref{T2.1} and \ref{T2.2} it follows that, for limited
sm-pairs, the set of all possible dynamic types coincides with the set $%
\Delta $ of infinite sequences $T_{i_{1}}T_{i_{2}}\cdots $ such that
$T_{i_{1}},T_{i_{2}},\ldots \in \{T_{2},\ldots ,T_{7}\}$ and $%
T_{i_{1}}\leqq T_{i_{2}}\leqq \cdots $. For arbitrary
sm-pairs the set of all possible dynamic types coincides with the set $%
\Delta \cup T_{1}^{\infty }$.

\section{Conclusions}
\label{S9}

In this paper, we studied computation trees over arbitrary structures. We described the set of all possible dynamic types for (i) arbitrary sm-pairs ((structure, measure)-pairs), (ii) for limited sm-pairs, and, in fact, (iii) for limited sm-pairs with a weighted depth as the complexity measure. The question about the set of all possible dynamic types for sm-pairs that have the depth as the complexity measure is open. In the future, we are planning to consider both this issue and various issues related to computational trees over structures with finite collections of predicates and functions.

\section*{Acknowledgments}

Research reported in this publication was supported by King Abdullah
University of Science and Technology (KAUST).

\bibliographystyle{spmpsci}
\bibliography{comp-trees}

\end{document}